\documentclass[a4paper,onecolumn,accepted=2024-05-22]{quantumarticle}
\pdfoutput=1

\usepackage[utf8]{inputenc}
\usepackage[english]{babel}
\usepackage[T1]{fontenc}

\usepackage[numbers,sort&compress]{natbib}

\usepackage{mathtools}
\usepackage{amsmath}
\usepackage{amsthm}
\usepackage{physics}
\usepackage{amsfonts}
\usepackage{hyperref}
\usepackage{verbatim}

\usepackage{algorithm2e}

\usepackage{caption}
\usepackage{subcaption}
\usepackage{graphicx}
\usepackage{float}
\usepackage{thmtools} 
\usepackage{thm-restate}

\usepackage[svgnames,dvipsnames]{xcolor}

\usepackage{stmaryrd}
\usepackage{empheq}
\usepackage{amssymb}

\theoremstyle{plain}
\newtheorem{theorem}{Theorem}
\newtheorem{corollary}{Corollary}
\newtheorem{proposition}{Proposition}

\theoremstyle{definition}
\newtheorem{definition}{Definition}
\newtheorem{remark}{Remark}

\usepackage{apptools}
\AtAppendix{\counterwithin{proposition}{section}}

\renewcommand\H{\mathcal H}

\renewcommand\S{\mathcal S}

\newcommand\C{\mathcal{C}}

\newcommand\M{\mathcal M}
\newcommand\W{\mathcal W}

\usepackage{dsfont}
\newcommand{\id}{\mathds{1}}

\usepackage[normalem]{ulem}

\begin{document}

\title{Improving social welfare in non-cooperative games with different types of quantum resources}

\author{Alastair A.\ Abbott}
\orcid{0000-0002-2759-633X}
\affiliation{Univ.\ Grenoble Alpes, Inria, 38000 Grenoble, France}

\author{Mehdi Mhalla}
\orcid{0000-0003-4178-5396}
\affiliation{Univ.\ Grenoble Alpes, CNRS, Grenoble INP, LIG, 38000 Grenoble, France}

\author{Pierre Pocreau}
\affiliation{Univ.\ Grenoble Alpes, Inria, 38000 Grenoble, France}
\affiliation{Univ.\ Grenoble Alpes, CNRS, Grenoble INP, LIG, 38000 Grenoble, France}

\maketitle

\begin{abstract}
We investigate what quantum advantages can be obtained in multipartite non-cooperative games by studying how different types of quantum resources can lead to new Nash equilibria and improve social welfare -- a measure of the quality of an equilibrium.
Two different quantum settings are analysed: a first, in which players are given direct access to an entangled quantum state, and a second, which we introduce here, in which they are only given classical advice obtained from quantum devices.
For a given game $G$, these two settings give rise to different equilibria characterised by the sets of equilibrium correlations $Q_\textrm{corr}(G)$ and $Q(G)$, respectively. 
We show that  $Q(G)\subseteq Q_\textrm{corr}(G)$, and by exploiting the self-testing property of some correlations, that the inclusion is strict for some games $G$. 
We make use of SDP optimisation techniques to study how these quantum resources can improve social welfare, obtaining upper and lower bounds on the social welfare reachable in each setting.
We investigate, for several games involving conflicting interests, how the social welfare depends on the bias of the game and improve upon a separation that was previously obtained using pseudo-telepathic solutions. 
\end{abstract}

\section{Introduction}

In game theory, the concept of \emph{Nash equilibria}~\cite{Nash} is used to define stable solutions of non-cooperative games in which players may have conflicting interests. It has been extensively studied for its many applications, from economics~\cite{myerson_nash_1999} to the analysis of information warfare~\cite{warfare}. 
In these games, each player has an individual payoff function, rewarding them depending on their behaviour and that of the others. 
A Nash equilibrium is a situation where no player has an incentive to deviate from their strategy and represents a strategically stable behaviour~\cite{gameeco}. 
One way to evaluate an equilibrium is to consider the mean (or total) payoff of the players, which is often referred to as its \emph{social welfare}.

A specific subset of such games in which the players' interests are aligned are commonly encountered in quantum information theory in the form of nonlocal games~\cite{brunner_bell_2014}. In such games, players sharing an entangled quantum state can achieve a higher probability of winning than if they shared any classical random variable.

\begin{figure}[ht]
    \centering
    \begin{tabular}{ccc}
        \includegraphics[width=0.3\textwidth]{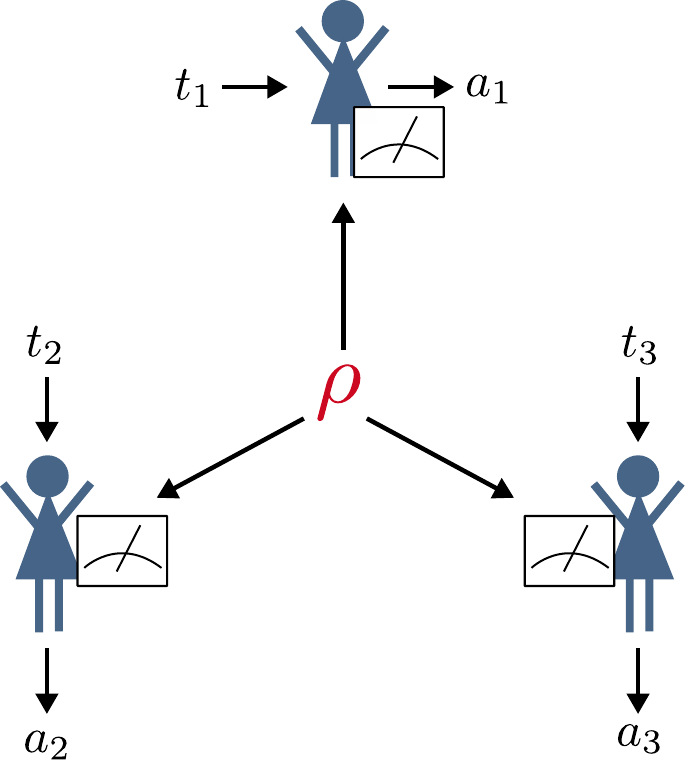} 
        & \hspace{0.1\textwidth} & \includegraphics[width=0.45\textwidth]{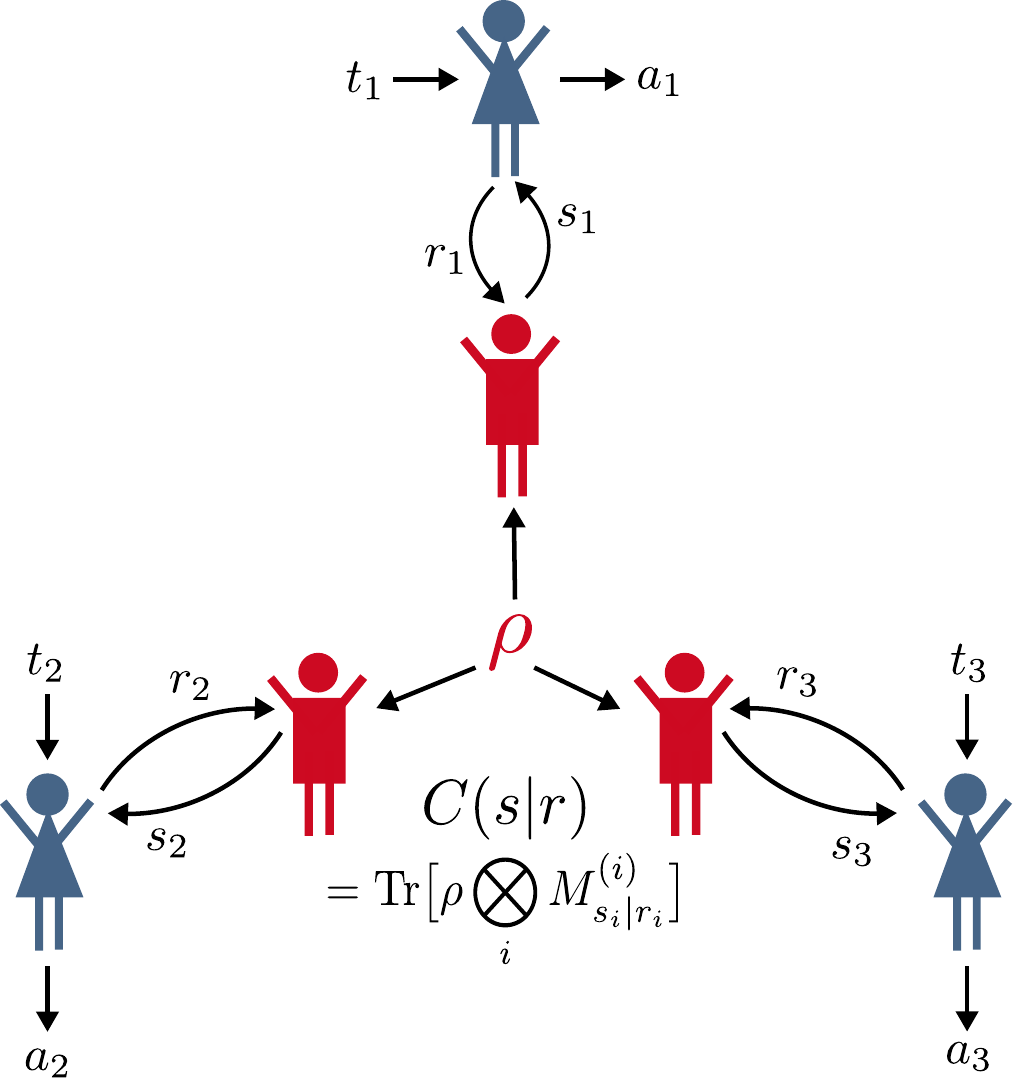} 
       \\
       (a) & \hspace{0.02\textwidth} & (b) 
    \end{tabular}
    \caption{The two quantum scenarios we consider. (a) ``Quantum advice'': direct access to a shared quantum state. (b) ``Quantum correlated advice'': classical access to a quantum correlation $C(s|r)$ via several local mediators sharing an entangled quantum state.}
    \label{fig:scenarios}
\end{figure}

In non-cooperative multipartite games with conflicting interest (involving more than two players) and with access to shared quantum states as resources~\cite{auletta_belief-invariant_2016}, the phenomena of entanglement and nonlocality are somewhat more complex~\cite{horodecki09,brunner_bell_2014}.
In this scenario, depicted in Figure~\ref{fig:scenarios}(a), it was recently shown that there exist three-player games where nonlocal resources
again provide advantages~\cite{bolonek-lason_three-player_2017}. 
In recent work, it was proved that by using pseudo-telepathic solutions on graph games and by allowing for multiple repetitions, it is possible to create an unbounded separation between classical and quantum social welfare~\cite{groisman_how_2020}. 

However, shared quantum states are just an example of an extra resource or \emph{advice} that can be distributed amongst the players for them to achieve better social welfare. 
In the study of non-cooperative games, several other types of classical resources have also historically been studied, including settings with shared random variables~\cite{AUMANN197467} or other types of correlated advice provided by a trusted mediator, such as belief-invariant (or non-signalling) advice~\cite{forges06,lehrer10}, or advice that allows information exchange between the players~\cite{auletta_belief-invariant_2016}.
Importantly, a desirable property of both shared random variables and shared quantum states is that they can be used to establish correlations amongst the actions of the players without the need for any trusted mediator. Because shared quantum states are more powerful than shared classical random variables, they allow the realisation of stronger belief-invariant correlations without the need for a mediator, motivating the use of quantum resources in these scenarios.

Here we introduce a new type of quantum resource, in the form of \emph{quantum correlated} advice, which gives the players classical access to a quantum correlation either through some black-box measuring devices shared amongst them (thus without the need of a centralised mediator), or indirectly, via one or several trusted mediators, who prepare and measure a distributed quantum state, as represented in the scenario of Figure~\ref{fig:scenarios}(b). 

This scenario, where players are given classical access to quantum resources, arises rather naturally. Indeed, it is likely that in the near future access to quantum resources will be restricted to those with sufficient means, such as large corporations. Players may therefore only have access to quantum states via local ``quantum centres'' (or mediators), which could share entangled states with other centres in a quantum network. In this paper we show that in conflicting-interest games, this setting (which is easier to implement) can in fact be more powerful than the situation where players have full quantum access (as studied in~\cite{auletta_belief-invariant_2016}), as it leads to more and potentially better equilibria.

More formally, for a game $G$ we define the sets of equilibrium correlations, $Q_{\text{corr}}(G)$ and $Q(G)$, which fully characterise the equilibria accessible when sharing respectively quantum correlated advice or a quantum state. We first show that for any game, $Q(G) \subseteq Q_{\text{corr}}(G)$. Then, to show that the two classes are, in general, strictly different, we exploit a property of some quantum correlations known as self-testing~\cite{Supic2020selftestingof}. This allows us, for a particular game of interest, to show that all strategies leading to some particular correlations are essentially equivalent up to local transformations, thereby reducing the problem of separating $Q(G)$ and $Q_{\text{corr}}(G)$ to the feasible task of showing that a particular solution -- inducing correlations in $Q_{\text{corr}}(G)$ -- is not a quantum equilibrium, proving that for some game $G$, $Q(G) \subsetneq Q_{\text{corr}}(G)$.

Finally, we compare the quality (in terms of social welfare) of the equilibria of these two quantum settings, using techniques from semidefinite programming to place upper and lower bounds on the achievable social welfare. 
Our results suggest, perhaps counterintuitively, that players can reach better equilibria when they do not have the freedom that comes with each having full access to their own quantum device. 
On the way, we show that relaxing the pseudo-telepathic constraints in some games allows the players to obtain higher social welfare, meaning that they can win more if they accept to lose in some rounds of the game.

The paper is organised as follows. 
In Section~\ref{sec:NC-eq} we recall the definition of non-cooperative multipartite binary games as well as the notions of equilibria and correlation resources. 
In Section~\ref{sec:q_strat} we introduce quantum equilibria as described by~\cite{auletta_belief-invariant_2016}, corresponding to situations where players share an entangled quantum state, before introducing the new class of quantum correlated advice we study in Section~\ref{sec:q_corr}. 
In Section~\ref{sec:inclusion} we study the relationship between these settings, and show how self-testing can be used to show a strict separation between them.
Finally, in Section~\ref{sec:social_welfare} we compare in detail three families of non-cooperative games, investigating the social welfare attainable in different settings.

\section{Non-cooperative games and quantum equilibria}
\label{sec:NC-eq}

\subsection{Non-cooperative multipartite games}

In an $n$-player multipartite game $G$, each player $i$ receives a type $t_i$ from a question $t = t_ 1\dots t_n$, and follows a strategy to produces an output $a_i$. The outputs are aggregated to form the answer $a = a_1\dots a_n$, and each player receives a payoff depending on the question and answer. 
More formally:
\begin{definition}
  An $n$ player non-cooperative binary game is defined by a set of valid questions $T \subseteq \{0, 1\}^n$, a prior probability distribution $\Pi$ over $T$ (satisfying $\forall t \in T$, $\Pi(t) \geq 0$ and $\sum_{t \in T} \Pi(t) = 1$), a set of valid answers  $A=A_1\times \cdots \times A_n = \{0, 1\}^n$, and an individual payoff function $u_i : A\times T \to \mathbb{R}$ for each player $i$.
\end{definition} 

Here, we will consider games where the payoff functions have the same form for all players and are defined from a set of ``winning'' input-output pairs $\W \subseteq A \times T$ as
\begin{equation} \label{eq:payoff}
    u_i(a, t) = \begin{cases} 0 & \text{ if } (a, t) \not\in \W, \\ 
    v_0 &\text{ if } a_i = 0 \text{ and } (a, t) \in \W, \\ 
    v_1 &\text{ if } a_i = 1 \text{ and } (a, t) \in \W, \end{cases}
\end{equation}
with $v_0,v_1>0$. 
Therefore, for a given question $t$, if two players output the same value they receive the same payoff.
It is the ratio $v_0/v_1$ that may create conflicts of interest between the players: if the payoff is unbalanced, e.g., if $v_1 > v_0$, the players might have an interest in losing more often if it means answering $1$ more often. The game thus becomes one of conflicting interests.

As an example, in Tables~\ref{table:NC5} and~\ref{table:NC5sym} we give the winning conditions of two (families of) $5$-player games defined in~\cite{groisman_how_2020}, and which we analyse in more detail in Section~\ref{sec:social_welfare}. 
These games are an adaption of multipartite binary graph games~\cite{anshu_contextuality_2020} which associate games to graphs and are based on the $C_5$ cycle graph; $\text{NC}_{01}(C_5)$ is a symmetrised version of $\text{NC}_{00}(C_5)$. 
Throughout this paper, we will always consider $\Pi$ to be the uniform distribution on $T$, so that the questions are drawn uniformly at random.

\begin{table}[ht]
\parbox{.45\linewidth}{
  \begin{tabular}{c | c}
    $\underset{t_1t_2t_3t_4t_5}{\text{Question}}$ & Winning condition \\
    \hline
    \hline
    10000 & $a_5 \oplus a_1 \oplus a_2 = 0$ \\
    01000 & $a_1 \oplus a_2 \oplus a_3 = 0$ \\
    00100 & $a_2 \oplus a_3 \oplus a_4 = 0$ \\
    00010 & $a_3 \oplus a_4 \oplus a_5 = 0$ \\
    00001 & $a_4 \oplus a_5 \oplus a_1 = 0$ \\
    11111 & $a_1 \oplus a_2 \oplus a_3 \oplus a_4 \oplus a_5 = 1$ \\
    \hline
    \end{tabular}
  \caption{Winning conditions for the $\text{NC}_{00}(C_5)$ family of games.}
  \label{table:NC5}
  }
  \hfill
  \parbox{.45\linewidth}{
  \begin{tabular}{c | c}
    $\underset{t_1t_2t_3t_4t_5}{\text{Question}}$ & Winning condition \\
    \hline
    \hline
    10100 & $a_5 \oplus a_1 \oplus a_2 = 0$ \\
    01010 & $a_1 \oplus a_2 \oplus a_3 = 0$ \\
    00101 & $a_2 \oplus a_3 \oplus a_4 = 0$ \\
    10010 & $a_3 \oplus a_4 \oplus a_5 = 0$ \\
    01001 & $a_4 \oplus a_5 \oplus a_1 = 0$ \\
    11111 & $a_1 \oplus a_2 \oplus a_3 \oplus a_4 \oplus a_5 = 1$ \\

    \hline
  \end{tabular}
    \caption{Winning conditions for the $\text{NC}_{01}(C_5)$ family of games.}
  \label{table:NC5sym}
  }
\end{table}

In a game, the strategy of a player influences their payoff but also that of other players. Naturally, each player seeks to maximise their own payoff. Some configurations of strategies are stable, meaning that no player can increase their payoff by unilaterally deviating from their strategy, and these configurations are called equilibria. Furthermore, when players are allowed to coordinate with some correlation or advice, they may reach different equilibria. In the next section formalise these concepts.

\subsection{Strategies and equilibria with correlation resources}

It is standard to consider a setting where the players can correlate their strategies via a shared correlation $C$ over two sets $R = R_1 \times \cdots \times R_n$ and $S = S_1 \times \cdots \times S_n$.
A standard interpretation, which we adopt here, is to see this correlation as that of a trusted mediator communicating with each party. 
Each party sends a message $r_i \in R_i$ to the mediator, who computes a correlation $C(s_1\dots s_n|r_1 \dots r_n)$ before sending back the advice $s_i \in S_i$ to each party (see Figure~\ref{fig:single_med}). 
Allowing the players to access, via the mediator, different types of correlations may give them different capabilities to obtain good equilibria. 

\begin{figure}[t]
    \centering
    \includegraphics[width=0.33\textwidth]{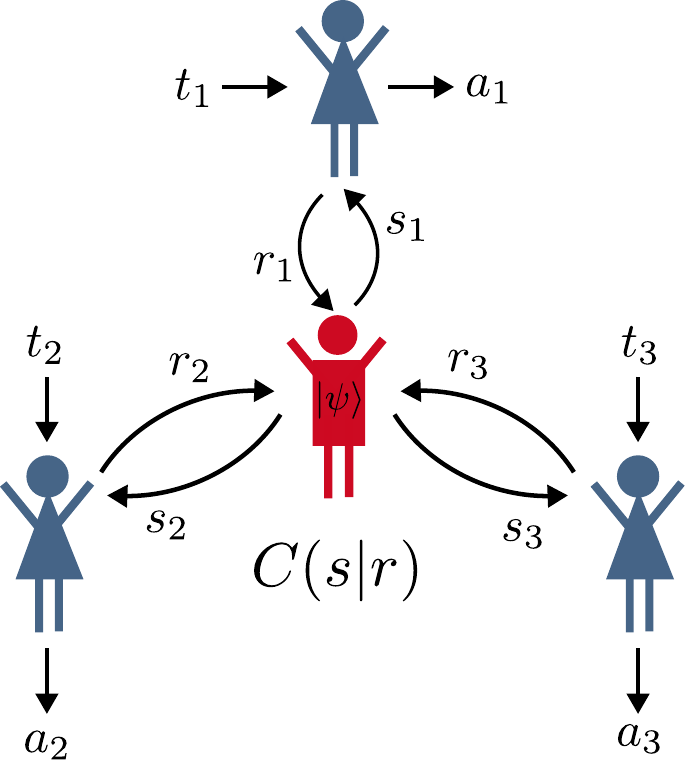}
    \caption{General scenario with correlation resources, in which players access a shared correlation $C(s|r)$ through communication with a trusted mediator.}
    \label{fig:single_med}
\end{figure}

We now formalise the notions of strategies and equilibria.
Throughout, for any string $t = t_1 \dots t_n$ we will use the shorthand notation that $t_{-i} = t_1\dots t_{i-1} t_{i+1} \dots t_n$ and $t_it_{-i} = t_{-i}t_i = t$. The simplest type of local strategy a player may use is a pure, or deterministic strategy.
\begin{definition}[Pure strategy]
  A pure strategy for a player $i$ is given by a pair of functions $f_i : T_i \to R_i$ and $g_i: T_i \times S_i \to A_i$, where $T_i = \{0, 1\}$, so that $T \subseteq T_1 \times \cdots \times T_n$. 
\end{definition}
More generally a player may use a probabilistic strategy obtained by mixing pure ones by marginalising over a set $\Lambda_i$ of a latent variables.
\begin{definition}[Mixed strategy]
  A mixed strategy for a player $i$ is given by a probability distribution $\pi_i : \Lambda_i \to \mathbb{R}$ and two functions $f_i:T_i\times \Lambda_i \to R_i$ and $g_i : T_i\times S_i \times \Lambda_i \to A_i$.
\end{definition}
We can now define solutions, which describe the strategies of each player as well as the shared correlations.
\begin{definition}[Solution]
  A solution is a tuple $(f, g, \pi, C)$ where $\pi = (\pi_1, \dots, \pi_n)$, $f = (f_1, \dots, f_n)$ and $g = (g_1, \dots, g_n)$ define the mixed strategies of each player $i$, and $C$ is the shared correlation. Then, a solution \emph{induces} a probability distribution on $A \times T$ as
  \begin{equation}\label{eq:induced_dist}
  	P(a|t) = \sum_{\lambda,s} C(s|f(t,\lambda))\,\pi(\lambda)\, \delta_{g(t,s,\lambda),a},
  \end{equation}
  where $\delta$ is the Kronecker delta.
  
  With a slight abuse of notation, we will use $(f, g, C)$ to denote the less general case of a \emph{pure solution}, with each player following a pure strategy. The expression of its induced probability distribution is then simplified to
  \begin{equation}
  	P(a|t) = \sum_{s} C(s|f(t)) \delta_{g(t,s),a}.
  \end{equation}
  Note that different solutions can induce the same probability distribution.
\end{definition}

A solution is a \emph{Nash equilibrium} if no player can improve their payoff by unilaterally deviating from their strategy. 
For a player $i$, a deviation is expressed by the use of a new mixed strategy. By linearity, however, it is sufficient to verify that for any fixed type $t_i$ the payoff decreases under any deviation to a pure strategy. That is, it is sufficient to check that it decreases for any alternative input $r_i$ sent by the player to the mediator, and any deterministic function $\mu_i : T_i \times S_i \to A_i$.

\begin{definition}[Nash equilibrium]
\label{def:NashEq} 
  A solution $(f, g, \pi, C)$ is a Nash equilibrium for a game $G$ if for all players $i$, for all $t_i$, $r_i$ and deterministic functions $\mu_i : T_i \times S_i \to A_i$:
  \begin{align}
  \label{ineq:nasheq}
  &\sum_{t_{-i}, \lambda, s} u_i(g(t_it_{-i}, s, \lambda),t_it_{-i})  C(s | f(t_it_{-i}, \lambda)) \Pi(t_it_{-i}) \pi(\lambda) \notag\\ 
  &\geq \sum_{t_{-i}, \lambda, s} u_i(\mu_i(t_i, s_i)g_{-i}(t_{-i}, s_{-i}, \lambda_{-i}),t_it_{-i}) C(s | r_i f_{-i}(t_{-i}, \lambda_{-i})) \Pi(t_it_{-i}) \pi(\lambda),
  \end{align}
  where we recall the convention that $t_i t_{-i} = t$. 
\end{definition}

The equilibria obtainable using correlated advice may depend on the characteristics of the correlation $C$. 
Thus, it is natural to distinguish different families of correlations, such as arbitrary correlations allowing players to exchange information about their types via the mediator, or local correlations, corresponding to a random variable shared between the players. 
These families and others will be introduced later. 
For a generic family of correlations $\C$, we introduce the notion of $\C$-Nash equilibrium.

\begin{definition}[$\C$-Nash equilibrium]
\label{def:CNash}
Let $\C$ be a family of correlations and $G$ a game. 
A $\C$-Nash equilibrium for $G$ is a solution $(f, g, \pi, C)$ that is a Nash equilibrium for $G$ and which has $C \in \C$.
\end{definition}

We furthermore introduce the concept of a canonical solution to simplify the manipulation of equilibria and ease calculations. 
The idea is to delegate all the ``computation'' of a strategy to the mediator, so that the players send directly their type to the mediator and simply output the advice they receive in return. Then, without loss of generality, the sets $T_i$ ($A_i$) and $R_i$ (resp.\ $S_i$) can be taken to be identical.

\begin{definition}[Canonical solution]
Let $(f,g, \pi, C)$ be a solution and $P$ its induced probability distribution on $A \times T$ as per Eq.~\eqref{eq:induced_dist}. Then, the pure solution $(id_{T},id_{A}, P)$ is called its canonical solution, where $id_{T}$ (resp.\ $id_{A})$ represents the identity function on $T$ (resp.\ $A$).
\end{definition}
By definition, a solution and its corresponding canonical solution both induce the same probability distribution $P$.
 
The conditions for a canonical solution to be a Nash equilibrium can be simplified using the facts that it is a pure solution and that we can identify each $T_i$ with $R_i$ and $A_i$ with $S_i$.

\begin{proposition}
\label{prop:canEq}
  A canonical solution $(id_{T}, id_{A}, P)$ is a Nash equilibrium for a game $G$, if and only if for all players $i$, all $t_i, r_i$, and all deterministic functions $\mu_i : T_i \times A_i \to A_i$:
\begin{equation}
    \sum_{t_{-i}, a} u_i(a,t_i t_{-i})  P(a | t_i t_{-i}) \Pi(t_i t_{-i}) \geq \sum_{t_{-i}, a} u_i(\mu_i(t_i, a_i)a_{-i},t_i t_{-i}) P(a | r_i t_{-i}) \Pi(t_i t_{-i}).
\end{equation}
\end{proposition}

A standard way to compare different equilibria is to consider the mean or total payoff of all the players~\cite{kaneko79} (we consider here the former measure). 
This quantity is called the \emph{social welfare}, and is obtained by considering the induced probability distribution of a given solution.

\begin{definition}[Social welfare]
For a game $G$, the \emph{social welfare} of a solution $(f, g, \pi, C)$ is defined as a function of its induced probability distribution $P$ by: 
\begin{equation} \label{eq:SW}
SW_G(P) = \frac{1}{n} \sum_{i}  \sum_{a, t} u_i(a, t)P(a|t)\Pi(t).
\end{equation}
\end{definition}

Since the social welfare is a function only of the induced probability distribution, we will sometimes simply speak of the social welfare of a probability distribution or correlation, leaving it implicit that there exists a solution inducing the probability distribution in question.
Note that two solutions inducing the same probability distribution therefore have the same social welfare. This is notably the case for a solution and its canonical form.

\subsection{Different types of correlations}
\label{eq:corr_families}

For a game $G$, the equilibria that players can access depend on the type of correlation resources they can share. 
Here we present several different families of correlations that have previously been studied in this context, and for each family $\C$ we consider the corresponding set of \emph{induced equilibrium correlations} associated to the $\C$-Nash equilibria for the class.
\begin{definition}[Set of induced equilibrium correlations] For a correlation family $\C$, the set of induced equilibrium correlations is the set of probability distributions
\begin{equation} \label{eq:eq_set}
	 \{P \ : \ \text{there exists a } \C\text{-Nash equilibrium } (f, g, \pi, C) \ \text{inducing $P$ as per Eq.~\eqref{eq:induced_dist}}\}.
\end{equation}
\end{definition}

Because the social welfare of a solution is characterised by its induced probability distribution, these sets are useful to compare the different types of advice and the quality of equilibria they allow one to reach.

The first set of correlations we introduce corresponds to a situation where the advice carries no information, or, put differently, when it can be factorised as a product of $n$-uncorrelated local distributions $L_i$. 

\begin{definition}[Factorable correlations] A factorable correlation $C$ is a probability distribution that can be written
\begin{equation}
    C(s|r)=L_1(s_1|r_1)\cdots L_n(s_n|r_n).
\end{equation}
\end{definition}

As such, the correlated advice can be absorbed into each player's local (mixed) strategy and this situation is equivalent to a setting without any correlated advice. The equilibria obtained are the basic Nash equilibria and the corresponding induced equilibrium correlations are denoted, for a given game $G$, $\text{Nash}(G)$.

The situation where players can share a correlated random variable corresponds to a family of correlations called local correlations.
\begin{definition}[Local correlations]
	A local correlation $C$ is a probability distribution that can be written
	\begin{equation}
		C(s|r) = \sum_{\gamma} V(\gamma) L_1(s_1|r_1,\gamma)L_2(s_2|r_2,\gamma)\cdots L_n(s_n|r_n,\gamma),
	\end{equation}
	where $V(\gamma)$ is a probability distribution over the shared random variable $\gamma$.
\end{definition}
In local correlations, any correlation between the players is solely due to $\gamma$; conditioned on its value, the advice a player receives from the mediator $L_i(s_i|r_i,\gamma)$ is uncorrelated from that of the other players. One can interpret this as the players receiving some correlated advice without the need to transfer any information about their type to the mediator, or simply having access to shared random variables in a scenario without mediator. 
Equilibria obtained using local correlations correspond to those intuitively obtainable using ``classical'' resources, and we denote by $\text{Corr}(G)$ the corresponding set of induced equilibrium correlations.\footnote{Note that this terminology is historical, and $\text{Corr}(G)$ denotes the set of equilibria obtained with \textit{local} correlations as advice which, is not the most general class of correlation.}

A more general family of correlations that has been considered in this context is that of belief invariant, or non-signalling, correlations~\cite{forges06,lehrer10}.
With these correlations, the action of the mediator may depend on the inputs of all parties, but in such a way that the knowledge of $(s_i, r_i)$ does not give information about the other inputs $r_j$. 
This can be interpreted as a privacy guarantee which can be valuable in conflicting interest games. 

To state these correlations, let us denote, for a subset $I \subset [n]:=\{1,\dots,n\}$, $R_I = \times_{i \in I} R_i$ (and likewise for $S_I$).
\begin{definition}[Belief invariant (non-signalling) correlations~\cite{forges06,kalai_how_2014}]
A correlation $C$ is belief invariant or non-signalling if, for any set of players $I \in [n]$, with $J = [n] \setminus I$ we have
\begin{equation}
    \forall s_I \in S_I,\ \forall r_I\in R_I, \ \forall r_J, r'_J \in R_J \quad \sum_{s_J \in S_J} C(s_I s_J | r_I r_J) = \sum_{s_J \in S_J} C(s_I s_J | r_I r'_J).
\end{equation}
\end{definition}
We denote the corresponding set of induced equilibrium correlations $\text{B.I.}(G)$.

Finally, the case where players are allowed, via the mediator, to share (or ``communicate'') arbitrary information about their type with the others corresponds to the situation where $C(s|r)$ may be an arbitrary probability distribution. Following~\cite{auletta_belief-invariant_2016}, we denote the corresponding induced equilibrium correlations $\text{Comm}(G)$. 

The inclusion structure of these families of correlations leads to a corresponding structure on the induced equilibrium correlations, namely~\cite{auletta_belief-invariant_2016,masanes_general_2006}:
\begin{equation}
	\text{Nash}(G) \subseteq \text{Corr}(G) \subseteq \text{B.I.}(G) \subseteq \text{Comm}(G).
\end{equation}

When a family of correlations $\C$ is closed under pre- and post-processing, the set of induced equilibrium correlations is generated by its canonical $\C$-Nash equilibria.
This is readily verified to be the case for all the classes we study in this paper.

\section{Quantum strategies}
\label{sec:q_strat}

The settings discussed above are classical insofar as the players have only classical capabilities: they interact with their advice classically and can essentially only choose to follow it, or ignore it and play some other strategy.

Quantum mechanics, however, allows for another natural possibility: for the players to receive intrinsically quantum advice.  

Several previous approaches have studied different ways of formalising this approach for non-cooperative games~\cite{la_mura_correlated_2005,10.1145/2090236.2090241,pappa_nonlocality_2015,auletta_belief-invariant_2016, Bostanci2022quantumgametheory} (see~\cite{khan2018quantum} for a review).
Here we follow the approach of~\cite{auletta_belief-invariant_2016} in which the players follow a fundamentally quantum strategy.
They each receive part of a multipartite quantum state as advice (e.g., provided again by an external mediator, or by the players interacting before playing the game), and can then perform measurements on the shared quantum system to correlate their actions (see Figure~\ref{fig:scenarios}(a)).

\begin{definition}[Quantum strategy]
A quantum strategy for a player $i$ is a set of positive-operator-valued measures (POVMs), $\M^{(i)} = \{M^{(i)}_{t_i}\}_{t_i \in T_i}$, where for each $t_i$ the POVM $M^{(i)}_{t_i} = \{M^{(i)}_{a_i|t_i}\}_{a_i \in A_i}$ satisfies $\sum_{a_i} M^{(i)}_{a_i | t_i} = \id$ and, for all $a_i,t_i$, $ M^{(i)}_{a_i|t_i} \geq 0$.
\end{definition}

A quantum solution thus consists of the POVMs of each player and the shared quantum state representing the advice.
\begin{definition}[Quantum solution]
A quantum solution $(\M, \rho)$ consists of a quantum state $\rho \in \S(\H_1 \otimes \dots \otimes \H_n)$ and sets of POVMs $\M = (\M^{(1)}, \dots, \M^{(n)})$ for every player. 
Each player has access to their part of $\rho$ in the corresponding Hilbert space $\H_i$. 
A quantum solution induces a probability distribution $P$ via the Born rule
\begin{equation}
    P(a|t) = \tr[\rho\,\left(M^{(1)}_{a_1|t_1}\otimes\cdots\otimes M^{(n)}_{a_n|t_n}\right)].
\end{equation}
  
\end{definition}

Crucially, compared to the classical solutions with correlated advice discussed previously, the notion of an equilibrium must be adjusted to take into account the quantum nature of the strategies and correspondingly greater freedom of players.
In this setting, a quantum solution is a \emph{quantum equilibrium} if no player can improve their payoff by changing their (local) choice of quantum strategy, i.e., by choosing different POVMs.

\begin{definition}[Quantum equilibrium]
 \label{def:Q}
 A quantum solution $(\M, \rho)$, is a \emph{quantum equilibrium} if for every player $i$, for every type $t_i$ and any POVM $N^{(i)} = \{N^{(i)}_{a_i}\}_{a_i \in A_i}$:
 
\begin{align}
  & \sum_{t_{-i}, a} u_{i}(a,t) \tr[\rho\,(M_{a_1|t_1}^{(1)}\otimes \cdots \otimes M_{a_n|t_n}^{(n)})]\Pi(t)  \\
  \geq & \sum_{t_{-i}, a} u_{i}(a,t) \tr[\rho\,(M_{a_1|t_1}^{(1)}\otimes \cdots \otimes M_{a_{i-1}|t_{i-1}}^{(i-1)}\otimes N_{a_{i}}^{(i)} \otimes M_{a_{i+1}|t_{i+1}}^{(i+1)} \otimes \cdots \otimes M_{a_{n}|t_{n}}^{(n)})] \Pi(t).\notag
\end{align}
\end{definition}

For a game $G$, we denote by $Q(G)$ the set of correlations induced by quantum equilibria. 
This set was previously studied in~\cite{auletta_belief-invariant_2016}, where it was compared to the types of classical advice presented in Section~\ref{eq:corr_families}.\footnote{Note, however, that Auletta et al.~\cite{auletta_belief-invariant_2016} denoted it $\text{Quantum(G)}$.}

Note that because no constraint is placed on the dimension of Hilbert spaces for strategies in defining $Q(G)$, we can without loss of generality assume the players' measurements are projective. 
Indeed, by taking a Naimark dilation of their POVMs and incorporating the local ancillary systems into the state, one obtains another quantum solution with projective measurements, producing the same correlations and which is an equilibrium if and only if the initial solution was one.
It is, however, not \emph{a priori} clear that one can consider only pure states. Indeed, while one can purify $\rho$ and obtain the same correlations, if $\rho$ is entangled it is no longer clear that the purified solution remains an equilibrium as the parties could perform measurements on the additional purification spaces.
Additionally, in contrast to the case of correlated advice, there is no clear notion of a canonical quantum solution. 

Let us finally note that, as also pointed out in~\cite{auletta_belief-invariant_2016}, one can verify that a solution is a quantum equilibrium using semidefinite programming. 
Indeed, we can maximise the average payoff of an individual player while keeping the strategies of the others, as well as the state $\rho$, fixed. 
If their mean payoff is already maximised by their current POVM $\{M^{(i)}_{a_i|t_i}\}_{a_i}$ for each input $t_i$, it means that they have no incentive to unilaterally change their strategy. 
This can be verified for each of the $n$ players separately.

\section{Quantum correlations}
\label{sec:q_corr}

In the type of quantum solutions introduced above, players share a potentially entangled quantum state and perform their own measurements. 
This supposes that they each have direct access to some local quantum devices. 
However, because of the complexity of such devices, it is also interesting to consider the case where they have classical access to multipartite quantum correlations, a situation which is more likely to appear in the near future. For instance, one could consider several spatially separated mediators sharing an entangled quantum state, each mediator communicating with one player as shown in Fig.~\ref{fig:scenarios}(b). We note that this provides a physical means to realise certain nonlocal belief invariant (i.e., no-signalling) correlations without requiring a single mediator to have access to information about all the players' types, which may be undesirable in many situations. Another way this scenario could arise is if the players access the correlation via black-box measuring devices which are provided to them before the start of the game. In either case, the players do not need any direct quantum capabilities.

This setting of classical access to quantum resources is extensively studied in device-independent approaches to quantum information~\cite{acin07,brunner_bell_2014}, and many quantum advantages or protocols can be adapted for it, such as the ability to perform blind quantum computation~\cite{broadbent_universal_2009}, where a player delegates their quantum computations to a server in a way that keeps their input, output and computation private.
Here we study how such resources are useful in non-cooperative games, and in particular in comparison to the fully-quantum strategies presented in the previous section.

Formally, this setting can be seen as a special case of the scenario with correlated advice provided by a single mediator depicted in Fig.~\ref{fig:single_med}, with $C$ being taken to be the set of $n$-partite quantum correlations, which can be obtained by performing individual measurements on each part of an $n$-partite quantum state.

\begin{definition}[Quantum correlations]
 \label{def:q_correlations} 
 A quantum correlation is a probability distribution $C(s|r)$ that can be written in the form

\begin{equation}
C(s|r) = \tr\big[\rho\,(M^{(1)}_{s_1|r_1} \otimes \dots \otimes M^{(n)}_{s_n | r_n})\big]
\end{equation}
for some quantum state $\rho \in S(\mathcal{H}_1 \otimes \dots \otimes \mathcal{H}_n)$ and POVMs $\M^{(i)}=\{M^{(i)}_{r_i}\}_{r_i\in R_i}$ with $M^{(i)}_{r_i}=\{M^{(i)}_{s_i|r_i}\}_{s_i\in S_i}$ in some Hilbert spaces $\mathcal{H}_i$.
\end{definition}

We denote the family of quantum correlations by $\C_Q$ and write the set of induced equilibrium correlations, as defined as in Eq.~\eqref{eq:eq_set}, as $Q_{\text{corr}}(G)$.

Because quantum mechanics obeys relativistic causality, it is well known that $n$-partite quantum correlations are non-signalling, and thus $Q_{\text{Corr}}(G) \subseteq \text{B.I.}(G)$. 
Since local correlations can be obtained from quantum correlations, we also have that $\text{Corr}(G) \subseteq Q_{\text{corr}}(G)$.
Finally, as shown by the following proposition, we can reduce the study of $\C_Q$-Nash equilibria to that of canonical $\C_Q$-Nash equilibria, allowing us to consider the simplified definition of Nash equilibrium for pure solutions as in Proposition~\ref{prop:canEq}.

\begin{proposition}
\label{prop:closed}
For any game $G$ we have $Q_\mathrm{corr}(G) \subseteq \C_Q$. Furthermore, for any probability distribution $P \in Q_\mathrm{corr}(G)$, the canonical solution $(id_T, id_A, P)$ is a $\C_Q$-Nash equilibrium.
\end{proposition}
\begin{proof}
Let $P \in Q_\textrm{corr}(G)$ be an induced equilibrium correlation. Then, by definition, there exists a $\C_Q$-Nash equilibrium $(f, g, \pi, C)$ inducing $P$. Because $C \in \C_Q$, there exist $(\rho, \M)$ such that for all $(s, r) \in S \times R$:
\begin{equation}
	C(s|r) = \tr \big[\rho\,(M^{(1)}_{s_1|r_1} \otimes \dots \otimes M^{(n)}_{s_n | r_n})\big].
\end{equation}
Consider then the new POVMs
\begin{equation}
	\label{eq:mixingPOVMs}
	\tilde{M}^{(i)}_{a_i|t_i} = \sum_\lambda M^{(i)}_{s_i|f_i(t_i, \lambda_i)} \pi(\lambda)\, \delta_{g(t,s,\lambda),a}.
\end{equation}
Then the induced distribution $P$ as defined in Eq.~\eqref{eq:induced_dist} can be written 
\begin{equation}
    P(a|t) = \tr\big[\rho\,(\tilde{M}^{(1)}_{a_1|t_1} \otimes \dots \otimes \tilde{M}^{(n)}_{a_n|t_n})\big],
\end{equation}
and thus $P \in \C_Q$.
Finally, because $(f, g, \pi, C)$ is a Nash equilibrium, it follows that $(id_T, id_A, P)$ is a $\C_Q$-Nash equilibrium.
\end{proof}

\section{Comparing different quantum resources in non-cooperative games}
\label{sec:inclusion}

We have introduced two quantum settings, one where the players each have their own quantum device and share an entangled quantum state, and another where the players receive part of a quantum correlation as advice, either by delegating their quantum devices to one or several mediators or by being afforded black-box access to quantum devices.
In the former setting, the players have quantum resources, while in the latter they can be considered to be fully classical, with only indirect access to quantum devices. 
For a game $G$, the Nash equilibria accessible in each setting are fully characterised by the induced equilibrium correlations $Q(G)$ and $Q_{\text{corr}}(G)$, respectively. They are both subsets of the quantum correlations $\C_Q$. 
Therefore, a natural way of comparing the two quantum settings is to compare $Q(G)$ and $Q_{\text{corr}}(G)$.

Naively, one may expect fully quantum players to be able to achieve more equilibria than players with only classical access to quantum devices.
However, the reality is actually precisely the opposite: the notion of a quantum equilibrium is more restrictive, since the type of deviations they can make (here, choosing an alternative local POVM), is more general, making it more difficult to obtain an equilibrium.
Formally, we have the following result.

\begin{theorem} 
\label{theorem:inclusion}
For any game $G$, $Q(G)\subseteq Q_{\mathrm{corr}}(G)$.
\end{theorem}
\begin{proof}
	We proceed by contradiction.
Let $C \in Q(G)$ be a correlation obtained by an equilibrium strategy $(\rho, \M)$, and suppose that $C \notin Q_{\text{corr}}(G)$. 
Then $C\in \C_Q$ and by assumption, the canonical solution $(id_T, id_A, C)$ is not a $\C_Q$-Nash equilibrium, so there exists a player $i$ who can increase their payoff by unilaterally deviating from their local strategy. 
However, this deviation can easily be expressed as a change in the POVMs in the quantum equilibrium solution, analogously to Eq.~\eqref{eq:mixingPOVMs} in the proof of Proposition~\ref{prop:closed}. 
Therefore, player $i$ in the quantum solution could also increase their payoff by unilaterally modifying their POVMs, which contradicts the fact that $(\rho, \M)$ is an equilibrium.
\end{proof}

Showing that for some games the above inclusion is strict would prove that the two settings are not equivalent and that players which have ``classical'' access to quantum devices can reach more equilibria. In the next subsection, we exhibit a game for which this is the case.

\subsection{Strict inclusion in a family of three-player games}
\label{sec:strict}

To find such an example, we introduce a family of three-player games for which the questions and winning conditions are summarised in Table~\ref{table:3players}. 
Like the games specified in Tables~\ref{table:NC5} and~\ref{table:NC5sym}, these games are a form of multipartite binary graph game~\cite{anshu_contextuality_2020}, this time based on the $C_3$ graph.
We assume a uniform distribution $\Pi$ over the questions and the payoff function is defined as in Eq.~\eqref{eq:payoff} from the corresponding winning set $\W$. 
The $\text{NC}(C_3)$ family is parametrised by the two values $v_0, \, v_1$ determining the precise payoff function, and each new pair $(v_0, v_1)$ specifies a particular game within this family.

\begin{table}
  \centering
  \begin{tabular}{c | c}
    $\underset{t_1t_2t_3}{\text{Question}}$ & Winning conditions \\
    \hline
    \hline
    100 & $a_1 \oplus a_2 \oplus a_3 = 0$ \\
    010 & $a_1 \oplus a_2 \oplus a_3 = 0$ \\
    001 & $a_1 \oplus a_2 \oplus a_3 = 0$ \\
    111 & $a_1 \oplus a_2 \oplus a_3 = 1$ \\
    \hline
  \end{tabular}
  \caption{Winning conditions for the $\text{NC}(C_3)$ games.}
  \label{table:3players}
\end{table}

A striking feature of these games is that they have a quantum \emph{pseudo-telepathic} solution allowing the players to win all the time (i.e., with probability 1), while this is impossible when they only have access to classical solutions using local correlations (or shared random variables) as advice~\cite{brassard_quantum_2005}. 
The pseudo-telepathic solution involves the players sharing a $C_3$ graph state $CZ^{(1,2)}CZ^{(2,3)}CZ^{(3,1)} \ket{+}^{\otimes3}$ (i.e., a GHZ-like state~\cite{GHZ}) and measuring the observables $Z$ or $X$ depending on whether their type is 0 or 1, respectively.
Writing the induced correlation of this quantum solution as $P$, then, for any value of $(v_0, v_1)$, the solution $(id_T, id_A, P)$ is a $\C_Q$-Nash equilibrium and $P \in Q_{\text{corr}}(G)$~\cite{groisman_how_2020}.
Similarly, one can readily check using semidefinite programming (as described in Sec.~\ref{sec:q_strat}) that this solution gives a quantum equilibrium and hence $P \in Q(G)$ as well. 
The social welfare of this solution is readily calculated to be $\frac{1}{2}(v_0 + v_1)$~\cite{groisman_how_2020}.

However, by modifying slightly the pseudo-telepathic solution, we can find solutions that still are $\C_Q$-Nash equilibria but which are themselves no longer quantum equilibria. 
To this end, let us take an angle $\theta \in [0, \pi]$ which parameterises the partially-entangled ``tilted'' state $\ket{\psi_\theta} = \cos(\frac{\theta}{2}) \ket{0} + \sin(\frac{\theta}{2}) \ket{1}$. 
In this ``tilted'' solution, the players now share the quantum state  $\ket*{\Psi_{\text{tilt}(\theta)}}=CZ^{(1,2)}CZ^{(2,3)}CZ^{(3,1)} (\ket{\psi_\theta}\otimes\ket{+}\otimes\ket{+})$, and the first player measures the observables $(X+Z)/ \sqrt{2}$ and $(X-Z)/ \sqrt{2}$ for each type $t_1=0$ and $t_1=1$ respectively, while the second and third players measure the observables $Z$ and $X$ for types $0$ and $1$, respectively. 
(Here $Z$ and $X$ are the corresponding Pauli observables, and the $+1$ and $-1$ measurement outcomes are interpreted as $0$ and $1$, respectively.)
We write the corresponding POVMs $\mathcal{M}_{\text{tilt(}\theta)}$, the correlation induced by this solution $C_{\text{tilt$(\theta)$}}$, and define $\rho_{\text{tilt(}\theta)} = \ketbra*{\Psi_{\text{tilt(}\theta)}}$. 

\begin{proposition}
\label{prop:tiltedSolution}
    For some values of $v_0$ and $v_1$ defining a $\textup{NC}(C_3)$ game $G$, there exist angles $\theta \in [0, \pi]$ for which the tilted solution induces a quantum correlated equilibrium but is noy a quantum equilibrium. That is, $C_\textup{tilt$(\theta)$} \in Q_\textup{corr}(G)$ but $C_\textup{tilt$(\theta)$} \not\in Q(G)$.
\end{proposition}

A simple example proving this proposition is obtained by taking $v_0 = v_1 = 1$ (a special case for which the players' interests are no longer in conflict) and $\theta = \pi/2$, for which the tilted solution corresponds to the pseudo-telepathic solution but with different measurements.
Because $C_\text{tilt$(\theta)$} \in \C_Q$, we know from Proposition~\ref{prop:closed} that $C_\text{tilt$(\theta)$} \in Q_\textrm{corr}(G)$ if and only if $(id_T, id_A, C_\text{tilt$(\theta)$})$ is a $\C_Q$-Nash equilibrium. This can readily be verified to be the case using the canonical Nash equilibrium conditions of Proposition~\ref{prop:canEq}.

The social welfare of this solution is $0.8536$, which is higher than that of the best classical solution (with a social welfare of $0.75$, see Figure~\ref{fig:results}) but lower than that of the pseudo-telepathic solution, which has a social welfare of $1$. By modifying their POVMs to that of the pseudo-telepathic solution, the first player can increase their payoff by $1-0.8536 = 0.1464$, showing that the tilted solution is indeed not a quantum equilibrium. 
We note that by varying the values of the parameters, one can find many more examples of such solutions and games; for less specific parameters, the quantum equilibrium property can be checked using SDPs as mentioned in Section~\ref{sec:q_strat}. 

It is important to stress that this argument is not, however, sufficient to show that $C_\text{tilt$(\theta)$} \not \in Q(G)$, as we did not rule out that some other quantum solution $(\tilde{\rho}, \tilde{\M})$ could induce the same correlation and be a quantum equilibrium.

\subsubsection{Self-testing of the tilted solution}

To demonstrate that the correlation $C_{\text{tilt(}\theta)}$ is not in $Q(G)$, we thus need to verify that no quantum equilibrium induces this correlation. 
However, determining all quantum solutions inducing a given probability distribution is a difficult task when one does not bound the dimension of the quantum systems being shared. 
In order to solve this problem, we will exploit methods from the self-testing of quantum systems~\cite{Supic2020selftestingof}.
Self-testing allows one to show that certain quantum correlations are only achievable by specific quantum solutions, up to local unitary transformations, and will allow us to show that any solution generating $C_{\text{tilt(}\theta)}$ is essentially equivalent to the tilted solution.

We begin by presenting the general definition of self-testing that we will build our approach upon, reformulating it in the language of solutions.

\begin{definition}[Self-testing of a solution~\cite{Supic2020selftestingof}]
\label{def:selftesting}
A correlation $C$ self-tests a solution $(\ketbra*{\psi}{\psi}, \M)$ if, for any solution $(\tilde\rho, \tilde{\M})$ that induces $C$ and any purification $\ket*{\tilde\psi}$ of $\tilde\rho$, there exists a local isometry $\Phi = \Phi_1 \otimes \dots \otimes \Phi_n$ and a state $\ket{\xi}$ such that
\begin{equation}
\Phi \big[\ket*{\tilde{\psi}} \big]  = \ket{\psi} \otimes \ket{\xi},
\end{equation}
and for any $1 \le i \le n$, $t_i \in T_i$ and $a_i \in A_i$
\begin{equation}
\Phi \big[\tilde{M}^{(i)}_{a_i|t_i} \ket*{\tilde{\psi}} \big]  = \big(M^{(i)}_{a_i|t_i}  \ket{\psi} \big) \otimes \ket{\xi}.
\end{equation}
\end{definition}

Note that we assume, as is standard in self-testing arguments~\cite{Supic2020selftestingof}, that $\rho$ is purified into an additional Hilbert space that is not accessible to any of the involved parties, upon which both the isometries $\Phi_i$ and POVMs $\mathcal{M}^{(i)}$ act trivially.
We also recall that, as discussed after Definition~\ref{def:Q}, we can without loss of generality assume that any solution we consider involves only projective measurements, and this fact is typically required in self-testing arguments.

The notion of self-testing in Definition~\ref{def:selftesting}, although powerful, is not sufficient for our purposes.
Indeed, we need to ensure that $(\tilde\rho,\tilde{\M})$ is a quantum equilibrium if and only if $(\ketbra{\psi},\M)$ is also one.
For this we need to be able to reason about the effect of any deviation (i.e., alternative POVM) a player may take.
To this end, we prove the following proposition about the tilted solution that constitutes a somewhat stronger notion of self-testing.

\begin{proposition}
\label{prop:selftest}
Let $(\tilde\rho, \tilde{\M})$ be a solution on some arbitrary Hilbert spaces which induces the distribution $C_{\textup{tilt$(\theta)$}}$ with $\theta \in (\frac{\pi}{4}, \, \frac{3\pi}{4})$, with $\tilde\M$ projective and $\ket*{\tilde\psi}$ a purification of $\tilde\rho$. 
For each party $i$ and $t_i\in T_i$, define also the observables $\tilde{A}^{(i)}_{t_i} = \tilde{M}^{(i)}_{0|t_i} - \tilde{M}^{(i)}_{1|t_i}$,
\begin{equation}
\tilde{X}_1 = \frac{\tilde{A}^{(1)}_0 + \tilde{A}^{(1)}_1}{\sqrt{2}}, \text{ and } \tilde{Z}_1 = \frac{\tilde{A}^{(1)}_0 - \tilde{A}^{(1)}_1}{\sqrt{2}},
\end{equation}
and for $i=2,3$
\begin{equation}
\label{eq:2tilde}
 \tilde{X}_i = \tilde{A}^{(i)}_1 \text{ and } \tilde{Z}_i = \tilde{A}^{(i)}_0.
\end{equation}
Then there exists a local isometry $\Phi = \Phi_1 \otimes \Phi_2 \otimes \Phi_3$ and state $\ket{\xi}$ such that
\begin{equation}
\Phi \big[\ket*{\tilde{\psi}} \big]  = \ket*{\Psi_{\textup{tilt}(\theta)}} \otimes \ket{\xi},
\end{equation}
and for $1 \le i \le 3$
\begin{equation}
\begin{aligned}
&\Phi \big[\tilde{X}_i \ket*{\tilde{\psi}} \big]  = \big(X_i  \ket*{\Psi_{\textup{tilt}(\theta)}} \big) \otimes \ket{\xi}, \\
&\Phi \big[\tilde{Z}_i \ket*{\tilde{\psi}} \big]  = \big(Z_i  \ket*{\Psi_{\textup{tilt}(\theta)}} \big) \otimes \ket{\xi},
\end{aligned}
\end{equation}
and thus the correlation $C_{\textup{tilt}(\theta)}$ self-tests the solution $(\rho_{\textup{tilt}(\theta)},\mathcal{M}_{\textup{tilt}(\theta)})$, where
\begin{equation}
	\rho_{\textup{tilt}(\theta)}=\ketbra*{\Psi_{\textup{tilt}(\theta)}}
\end{equation} 
Furthermore, we have that
\begin{equation}
\label{eq:selftestplus}
\Phi \big[\tilde{X}_i\tilde{Z}_i \ket*{\tilde{\psi}} \big]  = \big(X_iZ_i  \ket*{\Psi_{\textup{tilt}(\theta)}} \big) \otimes \ket{\xi}.
\end{equation}
\end{proposition}

Eq.~\eqref{eq:selftestplus} in particular goes beyond the standard notion of self-testing and is crucial here: it will allow us to map the action of any observable on the tilted state $\rho_{\text{tilt}(\theta)}$ to the action of an observable on the alternative state $\tilde{\rho}$ and prove Theorem~\ref{theorem:self-test} below. 
The proof of Proposition~\ref{prop:selftest} is a based on a tailored Bell inequality, following a method previously used in \cite{PhysRevLett.124.020402} to self-test a partially entangled GHZ state. The details of the proof are given in Appendix~\ref{appendix:self-testing}. 

We can now prove the following theorem.

\begin{theorem}
\label{theorem:self-test}
Let $G$ be any three-player game and $\theta \in (\frac{\pi}{4}, \, \frac{3\pi}{4})$. Then the correlation $C_{\textup{tilt}(\theta)}$ is in $Q(G)$ if and only if the tilted solution is a quantum equilibrium for $G$.
\end{theorem}
\begin{proof}
By definition, if the tilted solution $(\rho_{\text{tilt}(\theta)}, \M_{\text{tilt}(\theta)})$ is a quantum equilibrium for $G$, then its induced correlation $C_{\text{tilt}(\theta)}$ is in $Q(G)$.

To prove the converse, let us suppose for the sake of contradiction that $(\rho_{\text{tilt}(\theta)}, \M_{\text{tilt}(\theta)})$ is not a quantum equilibrium for $G$ but $C_\text{tilt$(\theta)$}$ is in $Q(G)$. 
Then there exists a player $i$, a type $t_i$, and a POVM $N^{(i)} = \{N^{(i)}_{a_i}\}_{a_i \in A_i}$ such that for any $a, t_{-i}$ (and recalling that here $n=3$)
\begin{align}\label{eq:bettersoln}
  & \sum_{t_{-i}, a} u_{i}(a,t) \tr[(M_{a_1|t_1}^{(1)}\otimes \cdots \otimes M_{a_n|t_n}^{(n)}) \, \rho_{\text{tilt}(\theta)}]\Pi(t)  \\
  < & \sum_{t_{-i}, a} u_{i}(a,t) \tr[(M_{a_1|t_1}^{(1)}\otimes \cdots \otimes M_{a_{i-1}|t_{i-1}}^{(i-1)}\otimes N_{a_{i}}^{(i)} \otimes M_{a_{i+1}|t_{i+1}}^{(i+1)} \otimes \cdots \otimes M_{a_{n}|t_{n}}^{(n)}) \, \rho_{\text{tilt}(\theta)}] \Pi(t).\notag
\end{align}

Because we assumed that $C_\text{tilt$(\theta)$}$ is an element of $Q(G)$, there exists a different solution $(\tilde\rho, \tilde{\M})$  inducing $C_{\text{tilt}(\theta)}$ and which is a  quantum equilibrium.
As previously argued we can assume $\tilde{\M}$ is projective, and we consider an arbitrary purification $\ket*{\tilde{\psi}}$ of $\tilde\rho$ (with the purifying space inaccessible to the players).
By Proposition~\ref{prop:selftest} there exists a local isometry $\Phi = \Phi_1 \otimes \Phi_2 \otimes \Phi_3$ and a state $\ket{\xi}$ such that
\begin{equation}
\Phi \big[\ket*{\tilde{\psi}} \big]  = \ket*{\Psi_{\text{tilt$(\theta)$}}} \otimes \ket{\xi},
\end{equation} 
and for any $1 \le i \le 3$:
\begin{align}
&\Phi \big[\tilde{X}_i \ket*{\tilde{\psi}} \big]  = \big(X_i  \ket*{\Psi_{\text{tilt$(\theta)$}}} \big) \otimes \ket{\xi},\label{eq:tiltisometryX} \\
&\Phi \big[\tilde{Z}_i \ket*{\tilde{\psi}} \big]  = \big(Z_i  \ket*{\Psi_{\text{tilt$(\theta)$}}} \big) \otimes \ket{\xi}, \\
&\Phi \big[\tilde{X}_i\tilde{Z}_i \ket*{\tilde{\psi}} \big]  = \big(X_iZ_i   \ket*{\Psi_{\text{tilt$(\theta)$}}}\big) \otimes \ket{\xi}.
\end{align}
Since each $N_{a_i}^{(i)}$ is a Hermitian matrix it can be written $N_{a_i}^{(i)} = \alpha\id_i + \beta X_i + \gamma Z_i + \epsilon \,\mathrm{i} X_iZ_i$, $\alpha, \beta,\gamma,\epsilon\in \mathbb{R}$ and $\mathrm{i}$ the imaginary unit.
Defining $\tilde{N}_{a_i}^{(i)}:=\alpha\tilde{\id}_i + \beta \tilde{X}_i + \gamma \tilde{Z}_i + \epsilon\,\mathrm{i} \tilde{X}_i\tilde{Z}_i$ we then have
\begin{equation}\label{eq:tiltisometryB}
\Phi\big[\tilde{N}_{a_i}^{(i)}\ket*{\tilde{\psi}}\big]=\Phi \big[\big( \alpha\tilde{\id}_i + \beta \tilde{X}_i + \gamma \tilde{Z}_i + \epsilon \,\mathrm{i} \tilde{X}_i\tilde{Z}_i \big) \ket*{\tilde{\psi}} \big]  = \big(N_{a_i}^{(i)}  \ket*{\Psi_{\text{tilt$(\theta)$}}} \big) \otimes \ket{\xi}
\end{equation}
and we obtain the following chain of inequalities:
\begin{align}
  & \sum_{t_{-i}, a} u_{i}(a,t) \tr[(\tilde{M}^{(1)}_{a_1|t_1}\otimes \cdots \otimes \tilde{M}^{(n)}_{a_n|t_n}) \, \tilde{\rho}]\Pi(t)  \\
  = & \sum_{t_{-i}, a} u_{i}(a,t) \tr[(M_{a_1|t_1}^{(1)}\otimes \cdots \otimes M_{a_n|t_n}^{(n)}) \, \rho_{\text{tilt}(\theta)}]\Pi(t)  \\
  < & \sum_{t_{-i}, a} u_{i}(a,t) \text{tr}\Big[(M_{a_1|t_1}^{(1)}\otimes \cdots \otimes M_{a_{i-1}|t_{i-1}}^{(i-1)}\otimes N_{a_{i}}^{(i)} \otimes M_{a_{i+1}|t_{i+1}}^{(i+1)} \otimes  \notag\\[-3mm]
  & \hspace{75mm} \otimes \cdots \otimes M_{a_{n}|t_{n}}^{(n)}) \, \rho_{\text{tilt}(\theta)} \otimes \ketbra*{\xi}\Big] \Pi(t) \\
  = & \sum_{t_{-i}, a} u_{i}(a,t)\, \tr[\Phi[(\tilde{M}^{(1)}_{a_1|t_1}\otimes \cdots \otimes \tilde{M}^{(i-1)}_{a_{i-1}|t_{i-1}}\otimes \tilde{N}^{(i)}_{a_{i}} \otimes \tilde{M}^{(i+1)}_{a_{i+1}|t_{i+1}} \otimes \cdots \otimes \tilde{M}^{(n)}_{a_{n}|t_{n}}) \, \tilde{\rho}]] \Pi(t) \\
  = & \sum_{t_{-i}, a} u_{i}(a,t) \tr[(\tilde{M}^{(1)}_{a_1|t_1}\otimes \cdots \otimes \tilde{M}^{(i-1)}_{a_{i-1}|t_{i-1}}\otimes \tilde{N}^{(i)}_{a_{i}} \otimes \tilde{M}^{(i+1)}_{a_{i+1}|t_{i+1}} \otimes \cdots \otimes \tilde{M}^{(n)}_{a_{n}|t_{n}}) \, \tilde{\rho}] \Pi(t),
\end{align}
where the second line follows from the fact that $(\rho_\text{tilt$(\theta)$},\mathcal{M}_{\text{tilt$(\theta)$}})$ and $(\tilde{\rho},\tilde{\mathcal{M}})$ induce the same distribution, the third line follows from Eq.~\eqref{eq:bettersoln} and the fact that $\tr[\ketbra*{\xi}]=1$, the fourth line from Eq.~\eqref{eq:tiltisometryB} and the fact that $\ket*{\tilde{\psi}}$ is a purification of $\tilde{\rho}$, and the final line from the conservation of the trace under isometries.
Hence, there exist a player $i$ who can improve their payoff with an alternative POVM $\tilde{N}^{(i)}$ and therefore $(\tilde{\rho}, \tilde{\M})$ is not a quantum equilibrium, which is a contradiction.
\end{proof}

Recall that Proposition~\ref{prop:tiltedSolution} showed that, for the game $G$ from the $\text{NC}(C_3)$ family with $v_0=v_1=1$, the tilted solution $(\rho_{\text{tilt}(\theta)},\mathcal{M}_{\text{tilt}(\theta)})$ with $\theta=\pi/2$ induces a quantum correlated equilibrium but is not itself a quantum equilibrium. 
Together with Theorem~\ref{theorem:self-test} this shows that $C_{\text{tilt}(\theta)}$ is not in $Q(G)$ and we thus arrive at the following corollary.
\begin{corollary}
    There exists a game $G$ for which $Q(G) \subsetneq Q_{\textup{corr}}(G)$.
\end{corollary}

As a consequence, we see that for some games, players with ``classical'' access to quantum resources can reach strictly more equilibria than those with direct access to the shared quantum state. In the next section we will use numerical methods to study whether the quality of the equilibria accessible in each setting also differs.

\begin{remark}
	Although we formulated, for clarity, Proposition~\ref{prop:selftest} and Theorem~\ref{theorem:self-test} explicitly for the tilted solution, we note that we could have formulated a somewhat more general result here. Indeed, Theorem~\ref{theorem:self-test} can be readily generalised to any correlation $C$ arising from a solution $(\rho,\M)$ as long as one can prove a generalised self-testing result analogous to Proposition~\ref{prop:selftest} linking $C$ and $(\rho,\M)$. In particular, one requires analogues of Eqs.~\eqref{eq:tiltisometryX}--\eqref{eq:tiltisometryB} for a tomographically complete set of operators, something which can be obtained from standard self-testing proofs in a wide range of scenarios~\cite{Laura-Private-Comm}.
\end{remark}

\section{Optimal social welfare in the two quantum settings}
\label{sec:social_welfare}

In Section~\ref{sec:inclusion} we proved that for some games, players with ``classical'' access to a quantum resource can reach strictly more equilibria than players with their own quantum devices. A natural question to ask is whether or not those equilibria only accessible in the first setting are actually better in term of social welfare than the others.  Indeed, even if $Q_\textrm{corr}(G)$ and $Q(G)$ differ, the \emph{quality} of attainable equilibria -- in terms of social welfare -- might not.
More generally, the social welfare is a more relevant measure of how useful different resources are in non-cooperative games.
By studying it in these two quantum settings, we gain better insight into the power of quantum resources depending on what kind of access is provided  -- direct quantum access, or indirect classical access.

In this section we therefore study the problem of optimising the social welfare over $\C_Q$-Nash equilibria and quantum equilibria solutions, or, put different, of optimising the social welfare over equilibrium correlations in $Q_\textrm{corr}(G)$ and $Q(G)$, respectively.
More formally, we wish to solve the optimisation problem
\begin{equation}
	\label{eq:opt_problem_general}
\max_P SW_G(P) = \frac{1}{n} \sum_{a, t} \sum_{i} u_i(a, t) P(a|t)\Pi(t),
\end{equation}
where the maximisation is either over $P\in Q_\textrm{corr}(G)$ or $P\in Q(G)$, respectively. 
We note that this maximisation problem is related to that of finding the maximal violation of a Bell inequality~\cite{Brunner2013, brunner_bell_2014}. Indeed, for games without any conflicting interests the two problems are equivalent as all correlations maximising the social welfare will necessarily be (both quantum and $\C_Q$-Nash) equilibria. 
However, in the case of conflicting-interest games, the correlations maximising the social welfare are not necessarily equilibria, and this difference will be apparent in the results we present below.

One can check directly whether a given distribution $P$ is in $Q_{\text{corr}}(G)$ by verifying the (finite) set of linear constraints defined by the Nash equilibrium condition of Proposition~\ref{prop:canEq}. 
However, there is no simple characterisation of the set $\C_Q$ of quantum correlations, let alone of $Q_\textrm{corr}(G)$, making solving exactly this optimisation problem difficult in this case.
Instead, here we make use of the NPA hierarchy~\cite{navascues_convergent_2008} which provides a convergent set of necessary conditions, each which can be expressed as a semidefinite program, leading to increasingly better approximations of the set of quantum correlations.

By optimising over distributions $P$ compatible with a given ``level'' of the NPA hierarchy and additionally imposing the linear constraints for $P$ to be a Nash equilibrium, we can obtain increasingly better upper bounds on the social welfare of correlations in $Q_{\text{corr}}(G)$ using semidefinite programming.

For the second setting, determining whether a quantum solution $(\rho,\M)$ is an equilibrium -- and thus whether the induced distribution $P\in Q(G)$ -- is itself an SDP, as discussed earlier.
The full optimisation problem of Eq.~\eqref{eq:opt_problem_general} nonetheless remains highly nonlinear, and moreover one can no longer simply impose the equilibrium conditions at the level of the correlations on the NPA hierarchy.
Finding nontrivial upper bounds on the social welfare of correlations in $Q(G)$ is thus a complicated problem, and remains one of the main barriers to proving a separation between the two settings.

Instead, we can study the tightness of the upper bounds we obtain for $Q_\textrm{corr}(G)$ and the potential separation between the settings by optimising over explicit quantum strategies.
In particular, we can use a see-saw iterative optimisation approach~\cite{werner_bell_2001}, alternating between optimising the social welfare over the state $\rho$ and the measurements of different individual parties.
More precisely, we can relax Eq.~\eqref{eq:opt_problem_general} as an optimisation over quantum solutions as
\begin{equation}
\label{seesaw}
    \max_{\M^{(1)}} \cdots \max_{M^{(n)}} \max_{\rho} SW_G(P) = \frac{1}{n} \sum_{a, t} \sum_{i} u_i(a, t) \tr[\rho\,(M^{(1)}_{a_1|t_1}\otimes\cdots\otimes M^{(n)}_{a_n|t_n})]\Pi(t).
\end{equation}
Starting from different, randomly chosen solutions (for a given dimension for each Hilbert space $\H_i$), we iteratively solve one maximisation problem at a time while keeping the other parts of the solution fixed.
Each iteration is itself an SDP, and while the procedure is not guaranteed to converge, in practice, by iterating sufficiently many times, such see-saw algorithms perform well. Note that while the algorithm optimises the social welfare over solutions whose induced probability distribution is in $\C_Q$, it does not guarantee that these solutions are equilibria in either $Q(G)$ or $Q_{\textrm{corr}}(G)$.

To ensure that an equilibrium condition is verified, this initial see-saw optimisation can be complemented by further rounds of iterative optimisation in which equilibrium constraints are imposed.
For $Q_{\textrm{corr}}(G)$, one can add directly the additional condition that no player can increase their payoff by applying a deterministic function to their output.
For $Q(G)$, such constraints, however, are themselves expressed as SDPs, and thus cannot be added to each round of optimisation.
Instead, one can modify the objective function so that each player optimises their own payoff. 
While this may lower the overall social welfare, it seeks to bring the solution towards a quantum equilibrium, and in practice we find the algorithm generally converges rapidly to such an equilibrium solution.

This approach is evidently not guaranteed to find a global maximum of the social welfare, but it provides lower bounds on the social welfare of correlations in $Q(G)$ and $Q_\textrm{corr}(G)$. 
By repeating it on a large number of randomly chosen initial points, we can gain significant insight into the attainable social welfare in the two different settings.

\subsection{Numerical results for three families of quantum games}

We apply these methods to three families of games: $\text{NC}(C_3)$ as defined in Table~\ref{table:3players}, as well as $\text{NC}_{00}(C_5)$ and $\text{NC}_{01}(C_5)$ whose questions and winning conditions are summarised in Tables~\ref{table:NC5} and~\ref{table:NC5sym}. Our code is available at~\cite{code}.
As elsewhere, we assume the priors $\Pi$ are uniform, and the payoff functions are defined as in Eq.~\eqref{eq:payoff}, thus paramaterising the families by $v_0$ and $v_1$.
Since the overall scale of the social welfare is unimportant and determined by $v_0 + v_1$, there is effectively one free parameter of interest, $v_0/(v_0+v_1)$, and in practice we fix $v_0+v_1=2$.

The lower bounds obtained via the see-saw approach are obtained by running the algorithm five times with five different, randomly chosen, initial solutions $(\rho, \M)$, where we took $\rho$ to be an $n$-qubit state, and keeping the highest social welfare amongst the results.
In most cases, the best solution we obtained is equivalent to that found by taking, as an initial seed, measurements such that each player measures in the computational basis when receiving the type $t_i = 0$ and in the Hadamard basis when $t_i = 1$.

The results obtained from the two approaches (upper bounds from the NPA hierarchy\footnote{We used a natural generalisation of the ``$1+AB$'' intermediate level of the hierarchy to consider all monomials of order 1 and all products of zero or one measurement operator from each party~\cite{navascues_convergent_2008}. E.g., for $n=3$ one might call this level ``$1+AB+AC+BC+ABC$''.} and see-saw lower bounds) for each family of games are shown in Figure~\ref{fig:results} as a function of $v_0/(v_0+v_1)$.
For reference, we also show the best social welfare attainable by ``classical'' equilibria, i.e., those obtainable with local correlations as advice and thus inducing equilibrium correlations in $\text{Corr}(G)$, as well as the social welfare of the pseudo-telepathic solutions on the range of parameters for which they remain quantum equilibria. 
For the five-player games, the pseudo-telepathic solution is the natural generalisation of the three-player one described in Section~\ref{sec:strict} \cite{anshu_contextuality_2020}.
The pseudo-telepathic solutions are are easily verified to be quantum equilibra when $v_0$ is high enough and to induce quantum correlated equilibria for the same range of $v_0/(v_0+v_1)$. 
They give a social welfare of 1 for all values of $v_0/(v_0+v_1)$.

\begin{figure}[ht]
     \centering
     \begin{subfigure}[b]{0.51\textwidth}
         \centering
         \includegraphics[width=\textwidth]{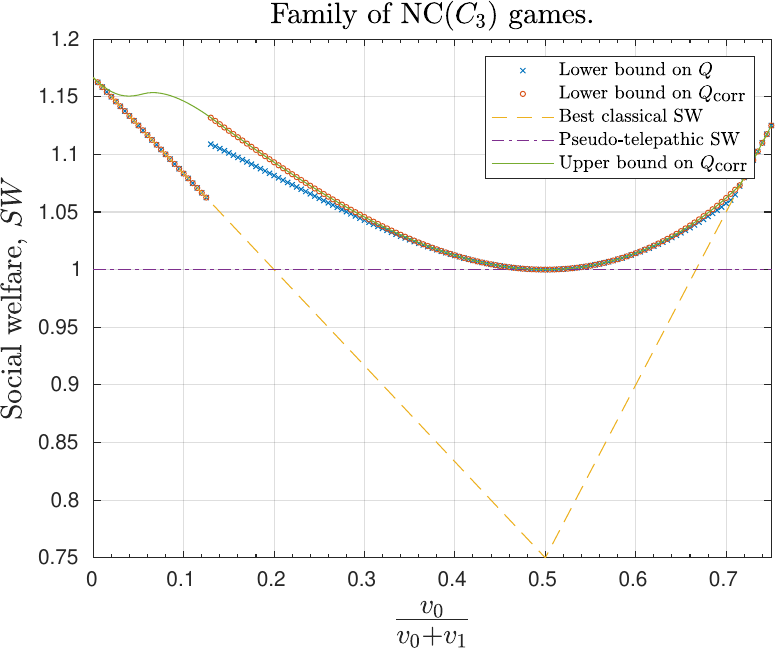}
         \label{fig:NC_C3}
     \end{subfigure}
     \hfill
     \begin{subfigure}[b]{0.49\textwidth}
         \centering
         \includegraphics[width=\textwidth]{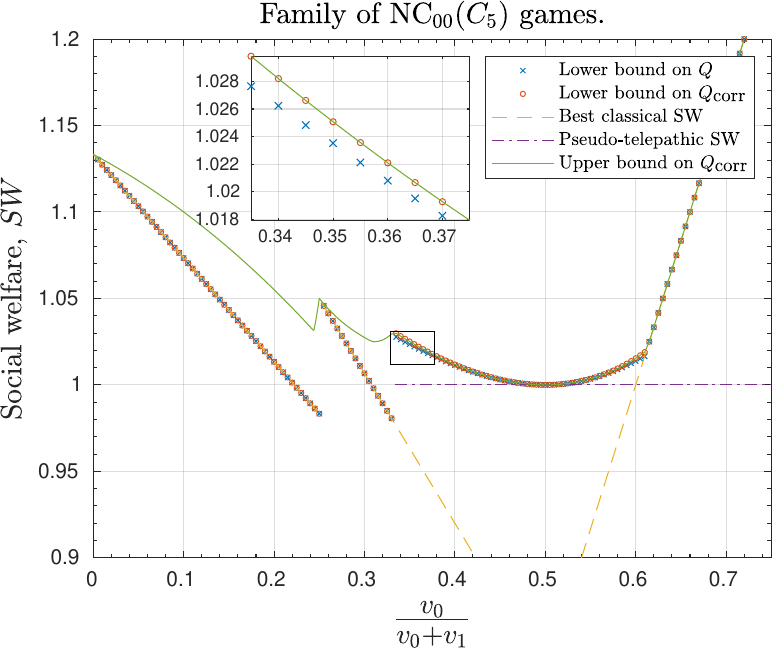}
         \label{fig:NC_C5sym}
     \end{subfigure}
     \hfill
     \begin{subfigure}[b]{0.49\textwidth}
         \centering
         \includegraphics[width=\textwidth]{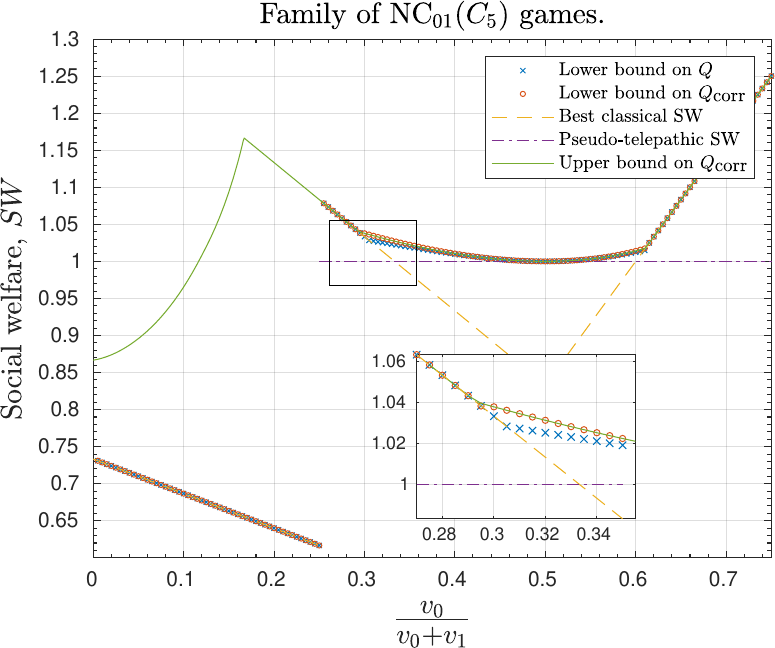}
         \label{fig:NC_C5}
     \end{subfigure}
            \caption{Numerically obtained upper and lower bounds on the social welfare attainable with quantum resources, as well as the social welfare of the best classical equilibria and of pseudo-telepathic solutions. The plots are truncated to the regions where the classical equilibria are not optimal.}
        \label{fig:results}
\end{figure}

Firstly, note that the social welfare of a given classical equilibrium solution is a linear function of $v_0$.
The changes of slope in the social welfare of optimal classical strategies correspond to a particular strategy no longer being an equilibrium for certain values of $v_0 / (v_0 + v_1)$.
For example, for the family $\text{NC}(C_3)$ the best classical equilibrium in the range $[0, 0.5]$ is $(id_T, id_A, 1)$ (and its symmetric equivalents), with $1$ representing the strategy where a player always answers $1$. 
However, this solution is no longer an equilibrium when the ratio $v_0/(v_0+v_1)$ goes above $0.5$, and the best equilibrium is then $(0,0,0)$, where $0$ represents the strategy of always answering $0$.

The first interesting feature we observe is that when the games are biased, i.e., when  $v_0\neq v_1$, the see-saw approach allows us to find explicit equilibria in both $Q_{\textrm{corr}}(G)$ and $Q(G)$ that have better social welfare than that of the pseudo-telepathic ones, showing that players can obtain a higher payoff if they accept losing the game sometimes.

In general, we can discern different behaviours for the upper and lower bounds, which we can separate into three regimes. 
In the extremes, when the ratio $v_0/(v_0 + v_1)$ is big enough, the best social welfare is reached by a classical equilibrium and is tight with the NPA upper bound, meaning that access to a quantum device of any kind cannot be used to improve the social welfare of the players. 

A second regime is observed in the region around $v_0/(v_0 + v_1) = 1/2$.
While at $v_0/(v_0 + v_1) = 1/2$ the pseudo-telepathic solution is optimal (and better than classical ones), as one moves away from this point, the see-saw approach provided equilibrium correlations in $Q_\textrm{corr}(G)$ that coincide with the upper bound from the NPA hierarchy, and with a social welfare better than that of the pseudo-telepathic solution.
We thus found, in this region, the best possible quantum correlated equilibrium solutions in terms of social welfare.
These solutions are not, however, quantum equilibria, and the optimal solutions found by the see-saw approach for $Q(G)$ were slightly lower, with a gap to those for $Q_\textrm{corr}(G)$ increasing the further one goes from $v_0/(v_0 + v_1) = 1/2$.
The maximal gaps we found for each game are $0.023$ for $G=\text{NC}(C_3)$ at $v_0/(v_0 + v_1) = 0.13$,  $0.0021$ for $G=\text{NC}(C_5)$ at $0.335$, and  $0.008$ for $G=\text{NC}_{01}(C_5)$ at $0.305$. 
These gaps support the idea that the the extra equilibria available to players in $Q_\text{corr}(G)$ indeed translate into better obtainable social welfare, although proving this rigorously -- e.g., by providing tighter upper bounds on $Q(G)$ -- remains an open problem.

The last regime appears when the ratio $v_0/(v_0+v_1)$ is small.
Then, the best solutions obtained with the see-saw approach collapse to the classical ones and a large gap is found between the social welfare of these explicit solutions and the upper bound from the NPA hierarchy. This raises the question of whether the classical solutions are optimal in this regime, and whether we might tighten the upper bound by increasing the level of the NPA hierarchy used to obtain the upper bounds. 
Furthermore, this regime illustrates the importance of considering equilibria in game theoretical settings, and stresses the difference between maximising the violation of a Bell inequality (seen as the social welfare of a game) and finding the equilibrium with maximal social welfare. Indeed, in the case of the symmetrised $\text{NC}_{01}(C_5)$ games, we observe that the upper bound on the maximal social welfare attainable in $Q_{\text{corr}}$ goes below $1$, whereas some quantum correlations could reach a value of at least $1$, the pseudo-telepathic solution being an explicit example obtaining the value 1.

\section{Conclusions and perspectives}

In this paper, we studied the utility of different types of quantum resources in non-cooperative games, both in terms of the sets of equilibrium correlations that can be obtained and in the social welfare of those equilibria.
In addition to finding, as has previously been observed~\cite{auletta_belief-invariant_2016,groisman_how_2020}, that quantum resources can indeed provide advantages beyond what can be done classically, we delineated two different types of quantum resources: those in which the players have direct access to a quantum state to measure, and those in which they have classical access to some quantum states, either via nonlocal boxes, or through the delegation of their measurements to one or several mediators.

To show a strict separation between the corresponding classes of equilibria -- $Q(G)$ and $Q_\text{corr}(G)$, respectively -- we employed techniques from self-testing. This novel use of a tool that has received significant recent attention~\cite{Supic2020selftestingof} allowed us to reduce different solutions producing the same correlation to a unique  solution whose equilibrium is to be tested, and thereby to show that $Q(G)$ is, perhaps surprisingly, strictly contained in $Q_\text{corr}(G)$.
We then extended and exploited SDP techniques to provide both upper and lower bounds on the social welfare of equilibria obtained using these different types of quantum resource.
Our numerical results show these generally provide good results and produced strong evidence that these different types of quantum resources translate into different observed phenomena in situations of conflicting interest.

Beyond the observation that quantum systems are a useful resource in ubiquitous non-cooperative games, these results raise several interesting lines of study. 
The separation between these two types of quantum resource, for example, provides a new potential angle towards understanding the limits and advantages of delegated quantum measurements, or between being provided with classical and quantum descriptions of a state.
It is likewise interesting to observe that one can, in conflicting-interest games, improve the social welfare by deviating from pseudotelepathic solutions, and investigating this further would be interesting. Indeed, in pseudotelepathic games players are guaranteed to obtain their expected payoff, whereas the best solutions we found provided a better average payoff, at the cost of requiring players to accept losing lots on some rounds; i.e., the variance in payoff is much higher.

Finally, a natural extension of this work is to consider a third intermediate setting where some players have their own quantum devices, while others have only classical access to some quantum resources.
This setting may be particularly relevant for applications in quantum cryptography, where an adversary may have different capabilities than the trusted parties.
In the scenario we considered, we found that even though players with quantum devices have more freedom and more strategies, their equilibria are generally worse than classical players with quantum correlated advice. 
However, in intermediate settings, it may be the case that owning a quantum device is an advantage.

\section*{Acknowledgements}

The authors would like to thank Ivan {\v{S}}upi{\'c}  and Laura  Man{\v{c}}inska for useful and helpful discussions about self-testing.
This work has been partially supported by the ANR project ANR-15-IDEX-02, and the Plan France 2030 projects ANR-22-CMAS-0001 (QuanTEdu-France) and ANR-22-PETQ-0007 (PEPR integrated project EPiQ).

\bibliographystyle{quantum}
\bibliography{refs}

\begin{thebibliography}{10}

\bibitem{Nash}
John~F. Nash.
\newblock ``Equilibrium points in $n$-person games''.
\newblock \href{https://doi.org/10.1073/pnas.36.1.48}{Proceedings of the
  National Academy of Sciences {\bf 36}, 48--49}~(1950).

\bibitem{myerson_nash_1999}
Roger~B. Myerson.
\newblock ``Nash equilibrium and the history of economic theory''.
\newblock \href{https://doi.org/10.1257/jel.37.3.1067}{Journal of Economic
  Literature {\bf 37}, 1067--1082}~(1999).

\bibitem{warfare}
Jorma Jormakka and Jarmo V.~E. M{\"o}ls{\"a}.
\newblock ``Modelling information warfare as a game''.
\newblock Journal of Information Warfare {\bf 4}, 12--25~(2005).
\newblock  url:~\url{https://www.jstor.org/stable/26504060}.

\bibitem{gameeco}
Robert Gibbons.
\newblock ``Game theory for applied economists''.
\newblock \href{https://doi.org/10.2307/j.ctvcmxrzd}{Princeton University
  Press}. Princeton, NJ, USA~(1992).

\bibitem{brunner_bell_2014}
Nicolas Brunner, Daniel Cavalcanti, Stefano Pironio, Valerio Scarani, and
  Stephanie Wehner.
\newblock ``Bell nonlocality''.
\newblock \href{https://doi.org/10.1103/RevModPhys.86.419}{Reviews of Modern
  Physics {\bf 86}, 419}~(2014).

\bibitem{auletta_belief-invariant_2016}
Vincenzo Auletta, Diodato Ferraioli, Ashutosh Rai, Giannicola Scarpa, and
  Andreas Winter.
\newblock ``Belief-invariant and quantum equilibria in games of incomplete
  information''.
\newblock \href{https://doi.org/10.1016/j.tcs.2021.09.041}{Theoretical Computer
  Science {\bf 895}, 151--177}~(2021).

\bibitem{horodecki09}
Ryszard Horodecki, Pawel Horodecki, Michal Horodecki, and Karol Horodecki.
\newblock ``Quantum entanglement''.
\newblock \href{https://doi.org/10.1103/RevModPhys.81.865}{Reviews of Modern
  Physics {\bf 81}, 865}~(2009).

\bibitem{bolonek-lason_three-player_2017}
Katarzyna Bolonek-Laso{\'n}.
\newblock ``Three-player conflicting interest games and nonlocality''.
\newblock \href{https://doi.org/10.1007/s11128-017-1635-6}{Quantum Information
  Processing {\bf 16}, 186}~(2017).

\bibitem{groisman_how_2020}
Berry Groisman, Michael {Mc Gettrick}, Mehdi Mhalla, and Marcin Paw{\l{}}owski.
\newblock ``How quantum information can improve social welfare''.
\newblock \href{https://doi.org/10.1109/JSAIT.2020.3012922}{IEEE Journal on
  Selected Areas in Information Theory {\bf 1}, 445--453}~(2020).

\bibitem{AUMANN197467}
Robert~J. Aumann.
\newblock ``Subjectivity and correlation in randomized strategies''.
\newblock \href{https://doi.org/10.1016/0304-4068(74)90037-8}{Journal of
  Mathematical Economics {\bf 1}, 67--96}~(1974).

\bibitem{forges06}
Fran{\c c}oise Forges.
\newblock ``Correlated equilibrium in games with incomplete information
  revisited''.
\newblock \href{https://doi.org/10.1007/s11238-006-9005-3}{Theory and Decision
  {\bf 61}, 329--344}~(2006).

\bibitem{lehrer10}
Ehud Lehrer, Dinah Rosenberg, and Eran Shmaya.
\newblock ``Signaling and mediation in games with common interests''.
\newblock \href{https://doi.org/10.1016/j.geb.2009.08.007}{Games and Economic
  Behavior {\bf 68}, 670--682}~(2010).

\bibitem{Supic2020selftestingof}
Ivan {\v{S}}upi{\'{c}} and Joseph Bowles.
\newblock ``Self-testing of quantum systems: a review''.
\newblock \href{https://doi.org/10.22331/q-2020-09-30-337}{{Quantum} {\bf 4},
  337}~(2020).

\bibitem{anshu_contextuality_2020}
Anurag Anshu, Peter H{\o}yer, Mehdi Mhalla, and Simon Perdrix.
\newblock ``Contextuality in multipartite pseudo-telepathy graph games''.
\newblock \href{https://doi.org/10.1016/j.jcss.2019.06.005}{Journal of Computer
  and System Sciences {\bf 107}, 156--165}~(2020).

\bibitem{kaneko79}
Mamoru Kaneko and Kenjiro Nakamura.
\newblock ``The {N}ash social welfare function''.
\newblock \href{https://doi.org/10.2307/1914191}{Econometrica {\bf 47},
  423--435}~(1979).

\bibitem{kalai_how_2014}
Yael~Tauman Kalai, Ran Raz, and Ron~D. Rothblum.
\newblock ``How to delegate computations: the power of no-signaling proofs''.
\newblock In Proceedings of the forty-sixth annual {ACM} symposium on {Theory}
  of computing.
\newblock \href{https://doi.org/10.1145/2591796.2591809}{Pages 485--494}.
\newblock New York, NY, USA~(2014). Association for Computing Machinery.

\bibitem{masanes_general_2006}
Lluis Masanes, Antonio Ac{\'{i}}n, and Nicolas Gisin.
\newblock ``General properties of nonsignaling theories''.
\newblock \href{https://doi.org/10.1103/PhysRevA.73.012112}{Physical Review A
  {\bf 73}, 012112}~(2006).

\bibitem{la_mura_correlated_2005}
Pierfrancesco La~Mura.
\newblock ``Correlated equilibria of classical strategic games with quantum
  signals''.
\newblock \href{https://doi.org/10.1142/S0219749905000724}{International
  Journal of Quantum Information {\bf 03}, 183--188}~(2005).

\bibitem{10.1145/2090236.2090241}
Shengyu Zhang.
\newblock ``Quantum strategic game theory''.
\newblock In Proceedings of the 3rd Innovations in Theoretical Computer Science
  Conference.
\newblock \href{https://doi.org/10.1145/2090236.2090241}{Pages 39--59}.
\newblock New York, NY, USA~(2012). Association for Computing Machinery.

\bibitem{pappa_nonlocality_2015}
Anna Pappa, Niraj Kumar, Thomas Lawson, Miklos Santha, Shengyu Zhang, Eleni
  Diamanti, and Iordanis Kerenidis.
\newblock ``Nonlocality and conflicting interest games''.
\newblock \href{https://doi.org/10.1103/PhysRevLett.114.020401}{Physical Review
  Letters {\bf 114}, 020401}~(2015).

\bibitem{Bostanci2022quantumgametheory}
John Bostanci and John Watrous.
\newblock ``Quantum game theory and the complexity of approximating quantum
  {N}ash equilibria''.
\newblock \href{https://doi.org/10.22331/q-2022-12-22-882}{{Quantum} {\bf 6},
  882}~(2022).

\bibitem{khan2018quantum}
Faisal~Shah Khan, Neal Solmeyer, Radhakrishnan Balu, and Travis~S Humble.
\newblock ``Quantum games: a review of the history, current state, and
  interpretation''.
\newblock \href{https://doi.org/10.1007/s11128-018-2082-8}{Quantum Information
  Processing {\bf 17}, 1--42}~(2018).

\bibitem{acin07}
Antonio Ac{\'{i}}n, Nicolas Brunner, Nicolas Gisin, Serge Massar, Stefano
  Pironio, and Valerio Scarani.
\newblock ``Device-independent security of quantum cryptography against
  collective attacks''.
\newblock \href{https://doi.org/10.1103/PhysRevLett.98.230501}{Physical Review
  Letters {\bf 98}, 230501}~(2007).

\bibitem{broadbent_universal_2009}
Anne Broadbent, Joseph Fitzsimons, and Elham Kashefi.
\newblock ``Universal blind quantum computation''.
\newblock In 2009 50th {Annual} {IEEE} {Symposium} on {Foundations} of
  {Computer} {Science}.
\newblock \href{https://doi.org/10.1109/FOCS.2009.36}{Pages 517--526}.
\newblock Atlanta, GA, USA~(2009). IEEE.

\bibitem{brassard_quantum_2005}
Gilles Brassard, Anne Broadbent, and Alain Tapp.
\newblock ``Quantum pseudo-telepathy''.
\newblock \href{https://doi.org/10.1007/s10701-005-7353-4}{Foundations of
  Physics {\bf 35}, 1877--1907}~(2005).

\bibitem{GHZ}
Daniel~M. Greenberger, Michael~A. Horne, and Anton Zeilinger.
\newblock ``Going beyond {B}ell's theorem''.
\newblock \href{https://doi.org/10.1007/978-94-017-0849-4_10}{Pages 69--72}.
\newblock Springer. Dordrecht, Netherlands~(1989).

\bibitem{PhysRevLett.124.020402}
Flavio Baccari, Remigiusz Augusiak, Ivan \ifmmode \check{S}\else
  \v{S}\fi{}upi\ifmmode~\acute{c}\else \'{c}\fi{}, Jordi Tura, and Antonio
  Ac\'{\i}n.
\newblock ``Scalable {B}ell inequalities for qubit graph states and robust
  self-testing''.
\newblock \href{https://doi.org/10.1103/PhysRevLett.124.020402}{Physical Review
  Letters {\bf 124}, 020402}~(2020).

\bibitem{Laura-Private-Comm}
Laura Man\v{c}inska.
\newblock Private communication~(2023).

\bibitem{Brunner2013}
Nicolas Brunner and Noah Linden.
\newblock ``Connection between {B}ell nonlocality and {B}ayesian game theory''.
\newblock \href{https://doi.org/10.1038/ncomms3057}{Nature Communications {\bf
  4}, 2057}~(2013).

\bibitem{navascues_convergent_2008}
Miguel Navascu{\'e}s, Stefano Pironio, and Antonio Ac{\'\i}n.
\newblock ``A convergent hierarchy of semidefinite programs characterizing the
  set of quantum correlations''.
\newblock \href{https://doi.org/10.1088/1367-2630/10/7/073013}{New Journal of
  Physics {\bf 10}, 073013}~(2008).

\bibitem{werner_bell_2001}
Reinhard~F. Werner and Michael~M. Wolf.
\newblock ``Bell inequalities and entanglement''.
\newblock \href{https://doi.org/10.26421/QIC1.3-1}{Quantum Information and
  Computation {\bf 1}, 1--25}~(2001).

\bibitem{code}
Alastair~A. Abbott, Mehdi Mhalla, and Pierre Pocreau.
\newblock
  code:~\href{https://github.com/pierrepocreau/Improving\_QSW}{pierrepocreau/Improving\_QSW}.

\end{thebibliography}

\appendix

\renewcommand{\theequation}{A\arabic{equation}}
\setcounter{equation}{0}

\section{Self-testing the tilted quantum solution}
\label{appendix:self-testing}

In this appendix we provide a detailed proof of the self-testing of the tilted quantum solution introduced in Section~\ref{sec:strict}. The proof follows closely the approach of~\cite{PhysRevLett.124.020402}.

First, let us recall that the tilted solution is defined for an angle $\theta \in [0,\pi]$ and is composed of the state $\ket*{\Psi_{\text{tilt}(\theta)}}=CZ^{(1,2)}CZ^{(2,3)}CZ^{(3,1)} (\ket{\psi_\theta}\otimes\ket{+}\otimes \ket{+})$ where $\ket{\psi_\theta} = \cos(\frac{\theta}{2}) \ket{0} + \sin(\frac{\theta}{2}) \ket{1}$ and the observables $A_0^{(1)} = (X_1+Z_1)/ \sqrt{2}$ and $A_1^{(1)} = (X_1-Z_1)/ \sqrt{2}$ as well as $A_0^{(i)} = Z_i$ and $A_1^{(i)} = X_i$ for $i = 2, 3$. The associated POVMs will be written $\mathcal{M}_{\text{tilt(}\theta)}$ and the correlations generated on the tilted state, $C_{\text{tilt$(\theta)$}}$.
We restate Proposition~\ref{prop:selftest} for convenience before giving the proof.
\begin{proposition}
Let $(\tilde\rho, \tilde{\M})$ be a solution on some arbitrary Hilbert spaces which induces the distribution $C_{\textup{tilt$(\theta)$}}$ with $\theta \in (\frac{\pi}{4}, \, \frac{3\pi}{4})$, with $\tilde\M$ projective and $\ket*{\tilde\psi}$ a purification of $\tilde\rho$. 
For each party $i$ and $t_i\in T_i$, define also the observables $\tilde{A}^{(i)}_{t_i} = \tilde{M}^{(i)}_{0|t_i} - \tilde{M}^{(i)}_{1|t_i}$,
\begin{equation}
\tilde{X}_1 = \frac{\tilde{A}^{(1)}_0 + \tilde{A}^{(1)}_1}{\sqrt{2}}, \text{ and } \tilde{Z}_1 = \frac{\tilde{A}^{(1)}_0 - \tilde{A}^{(1)}_1}{\sqrt{2}},
\end{equation}
and for $i=2,3$
\begin{equation}
\label{eq:2tildeApp}
 \tilde{X}_i = \tilde{A}^{(i)}_1 \text{ and } \tilde{Z}_i = \tilde{A}^{(i)}_0.
\end{equation}
Then there exists a local isometry $\Phi = \Phi_1 \otimes \Phi_2 \otimes \Phi_3$ and state $\ket{\xi}$ such that
\begin{equation}
\Phi \big[\ket*{\tilde{\psi}} \big]  = \ket*{\Psi_{\textup{tilt}(\theta)}} \otimes \ket{\xi},
\end{equation}
and for $1 \le i \le 3$
\begin{equation}
\begin{aligned}
&\Phi \big[\tilde{X}_i \ket*{\tilde{\psi}} \big]  = \big(X_i  \ket*{\Psi_{\textup{tilt}(\theta)}} \big) \otimes \ket{\xi}, \\
&\Phi \big[\tilde{Z}_i \ket*{\tilde{\psi}} \big]  = \big(Z_i  \ket*{\Psi_{\textup{tilt}(\theta)}} \big) \otimes \ket{\xi},
\end{aligned}
\end{equation}
and thus the correlation $C_{\textup{tilt}(\theta)}$ self-tests the solution $(\rho_{\textup{tilt}(\theta)},\mathcal{M}_{\textup{tilt}(\theta)})$, where
\begin{equation}
	 \rho_{\textup{tilt}(\theta)}=\ketbra*{\Psi_{\textup{tilt}(\theta)}}.
\end{equation}
Furthermore, we have that
\begin{equation}
\label{eq:selftestplusApp}
\Phi \big[\tilde{X}_i\tilde{Z}_i \ket*{\tilde{\psi}} \big]  = \big(X_iZ_i  \ket*{\Psi_{\textup{tilt}(\theta)}} \big) \otimes \ket{\xi}.
\end{equation}
\end{proposition}

\begin{proof}
To prove the self-testing statement, we begin by giving a Bell inequality $\mathcal{I}_\theta \le \beta_C$ that is maximally violated by the tilted solution, adapting a method detailed in \cite{PhysRevLett.124.020402} to self-test partially entangled GHZ states. 

To tailor a Bell inequality with this property, let us consider the three stabilisers of the tilted state,
\begin{equation}
\label{eq:stab1}
S_1 = \sin(\theta)X_1Z_2Z_3 + \cos(\theta)Z_1,
\end{equation}
and, for $i = 2,3$,
\begin{equation}
\label{eq:stabi}
S_i = Z_{i-1}X_iZ_{i+1},
\end{equation}
where addition is performed modulo $3$, and the identity is left implicit on unspecified systems.
We set ourselves in a general situation where three players measure a physical tripartite state with the $\pm 1$-valued observables $\tilde{A}^{(i)}_0$ and $\tilde{A}^{(i)}_1$ in some arbitrary Hilbert spaces and define
\begin{equation}
\label{eq:1tilde}
\tilde{X}_1 = \frac{\tilde{A}^{(1)}_0 + \tilde{A}^{(1)}_1}{\sqrt{2}},\quad \tilde{Z}_1 = \frac{\tilde{A}^{(1)}_0 - \tilde{A}^{(1)}_1}{\sqrt{2}},
\end{equation}
and, for $i=2,3$,
\begin{equation}
\label{eq:2tilde-app}
 \tilde{X_i} = \tilde{A}^{(i)}_1 \text{ and } \tilde{Z_i} = \tilde{A}^{(i)}_0.
\end{equation}
In analogy to \eqref{eq:stab1} and~\eqref{eq:stabi}, we define the operators $\tilde{S_1} = \sin(\theta)\tilde{X}_1\tilde{Z}_2\tilde{Z}_3 + \cos(\theta)\tilde{Z}_1$ and, for $i= 2, 3$, $\tilde{S_i} = \tilde{Z}_{i-1}\tilde{X}_i\tilde{Z}_{i+1}$.

Let us now consider the Bell operator $\mathcal{B}_\theta$ corresponding to the (yet to be specified) Bell functional $\mathcal{I}_\theta$, with maximum quantum value $\beta_Q$.
We would like the shifted Bell operator $\beta_Q \id - \mathcal{B}_\theta$ to admit a sum of squares decomposition of the form
\begin{equation}
\label{eq:SOS}
\beta_Q \id - \mathcal{B}_\theta = \sum_{i=1}^3 \alpha_i^2(\id - \tilde{S_i})^2.
\end{equation}
By computing the squares of the $\tilde{S}_i$'s we obtain
\begin{equation}
\tilde{S_1}^2 = \id + \frac{1}{2}\biggl[\sin^2(\theta) - \cos^2(\theta)\biggl] \biggl\{\tilde{A}_0^{(1)}, \tilde{A}_1^{(1)}\biggl\},
\end{equation}
and
\begin{equation}
\tilde{S_2}^2 = \tilde{S_3}^2 = \id - \frac{1}{2} \biggl\{\tilde{A}_0^{(1)}, \tilde{A}_1^{(1)}\biggl\}, 
\end{equation}
where $\{\cdot,\cdot\}$ denotes the anticommutator and we have used the fact that each $A_{t_i}^{(i)}$ squares to the identity, which follows from the fact that the POVMs $M_{t_i}^{(i)}$ are projective. We thus find that
\begin{equation}
\begin{split}
\sum_{i=1}^3 \alpha_i^2(\id - \tilde{S_i})^2 =& \sum_{i=1}^3 \alpha_i^2  \id - 2\sum_{i=1}^3 \alpha_i^2 \tilde{S_i}\\ 
&+ \sum_{i=1}^3 \alpha_i^2  \id  
+ \frac{1}{2}\biggl[\big(\sin^2(\theta) - \cos^2(\theta)\big) \alpha_1^2 - (\alpha_2^2 + \alpha_3^2)\biggl] \biggl\{\tilde{A}_0^{(1)}, \tilde{A}_1^{(1)}\biggl\}.
\end{split}
\end{equation}
For $\theta \in (\pi/4, 3\pi/4)$ the final term can be eliminated by setting $\alpha_2^2 = \alpha_3^2 = \sqrt{2}$ and  $\alpha_1^2 = \frac{2\sqrt{2}}{\sin^2(\theta) - \cos^2(\theta)}$, to obtain
\begin{equation}
	\label{eq:SOSdecomp}
\begin{split}
\sum_{i=1}^3 \alpha_i^2(\id - \tilde{S_i})^2 &= \left(\frac{4\sqrt{2}}{\sin^2(\theta) - \cos^2(\theta)} + 4\sqrt{2}\right) \id - \left(\frac{4\sqrt{2}}{\sin^2(\theta) - \cos^2(\theta)} \tilde{S_1} + 2\sqrt{2}(\tilde{S_2}+\tilde{S_3})\right). 
\end{split}
\end{equation}
This then indeed has the desired form of the shifted Bell operator~\eqref{eq:SOS}, where can identify the Bell operator as
\begin{align}
\label{eq:bell}
\mathcal{B}_\theta =& \frac{4\sqrt{2}}{\sin^2(\theta) - \cos^2(\theta)} \tilde{S_1} + 2\sqrt{2}(\tilde{S_2}+\tilde{S_3})\\
=& \frac{4}{\sin^2(\theta) - \cos^2(\theta)} \biggl[\sin(\theta)\big(\tilde{A}_0^{(1)} + \tilde{A}_1^{(1)}\big)\tilde{A}_0^{(2)}\tilde{A}_0^{(3)} + \cos(\theta)\big(\tilde{A}_0^{(1)} - \tilde{A}_1^{(1)}\big) \biggl] \notag\\
&+  2\big(\tilde{A}_0^{(1)} - \tilde{A}_1^{(1)}\big)\tilde{A}_1^{(2)}\tilde{A}_0^{(3)} +  2\big(\tilde{A}_0^{(1)} - \tilde{A}_1^{(1)}\big)\tilde{A}_0^{(2)}\tilde{A}_1^{(3)},
\end{align}
and hence deduce the tailored Bell functional to be
\begin{align}
\mathcal{I}_{\theta} =& \frac{4}{\sin^2(\theta) - \cos^2(\theta)} \biggl[\sin(\theta)\big<\big(\tilde{A}_0^{(1)} + \tilde{A}_1^{(1)}\big)\tilde{A}_0^{(2)}\tilde{A}_0^{(3)}\big> + \cos(\theta)\big<\big(\tilde{A}_0^{(1)} - \tilde{A}_1^{(1)}\big)\big> \biggl] \notag\\
&+  2\big<\big(\tilde{A}_0^{(1)} - \tilde{A}_1^{(1)}\big)\tilde{A}_1^{(2)}\tilde{A}_0^{(3)}\big> +  2\big<\big(\tilde{A}_0^{(1)} - \tilde{A}_1^{(1)}\big)\tilde{A}_0^{(2)}\tilde{A}_1^{(3)}\big>.
\end{align}
The sum of squares decomposition of the shifted Bell operator implies that it is positive semi-definite, and therefore we can deduce from Eq.~\eqref{eq:SOSdecomp} that 
\begin{equation}
\beta_Q = 4\sqrt{2} \biggl[ \frac{1}{\sin^2(\theta)-\cos^2(\theta)} + 1\biggl]
\end{equation}
is an upper bound on the maximal quantum value $\beta_Q$ of $\mathcal{I}_\theta$.
In fact this upper bound is tight: one can easily verify that it is saturated by the tilted solution, i.e., $\bra*{\Psi_{\text{tilt}(\theta)}}\mathcal{B}_\theta \ket*{\Psi_{\text{tilt}(\theta)}} = \beta_Q$. Indeed, this follows by construction.

The classical bound $\beta_C$ of the Bell inequality can also be readily obtained by assigning a value $\pm 1$ to each observable, so that either $\tilde{A}_0^{(1)} + \tilde{A}_1^{(1)} = \pm2$ and $\tilde{A}_0^{(1)} - \tilde{A}_1^{(1)} = 0$, or the other way around. For $\theta \in (\pi/4, 3\pi/4)$ the maximal classical value is therefore
\begin{equation}
\beta_C = \begin{cases}
\frac{4\cos(\theta)}{\sin^2(\theta) - \cos^2(\theta)} + 4, & \text{for  $\theta \leq \pi/2$} \\
\frac{4\sin(\theta)}{\sin^2(\theta) - \cos^2(\theta)}, & \text{otherwise.}
\end{cases}
\end{equation}

Let us now consider a state $\ket*{\tilde\psi}$ and some observables $\tilde{A}^{(i)}_0$ and $\tilde{A}^{(i)}_1$ which achieve the maximal quantum value $\beta_Q$ of $\mathcal{I}_\theta$. To find an isometry that maps this state and measurements to the tilted quantum solution, we will first show that the $\tilde{X}_i,\tilde{Z}_i$ act like the Pauli $X$ and $Z$ on the tilted state, meaning that that those two operators square to the identity and anticommute when applied to the state $\ket*{\tilde\psi}$.
Note that by definition
\begin{equation}
\label{eq:anticom}
\{\tilde{X}_1, \tilde{Z}_1\} = 0,
\end{equation}
and that for $i=2,3$
\begin{equation}
\label{eq:squarei}
\tilde{Z}_i^2 = \tilde{X}_i^2 = \id.
\end{equation}
It thus remains to show that $\tilde{Z}_1^2 \ket*{\tilde\psi} =\tilde{X}_1^2 \ket*{\tilde\psi}= \ket*{\tilde\psi}$ and $\{\tilde{X}_i,\tilde{Z}_i\} \ket*{\tilde\psi}=0$ for $i=2,3$.

Notice firstly that because $\ket*{\tilde\psi}$ achieves the maximal quantum value $\beta_Q$, Eq.~\eqref{eq:SOS} implies that the $\tilde{S}_i$'s are stabilisers of $\ket*{\tilde\psi}$. 
Then, combined with the above results, we can show that $\tilde{Z}_1^2$ acts as the identity on $\ket*{\tilde{\psi}}$,
\begin{equation}
	\label{eq:Z1_id}
	\tilde{Z}_1^2\ket*{\tilde{\psi}} = \tilde{Z}_1^2 \tilde{X}_2^2\tilde{Z}_3^2\ket*{\tilde{\psi}} = \tilde{S}_2^2\ket*{\tilde{\psi}} = \ket*{\tilde{\psi}}.
\end{equation}
To prove the analogous result for $\tilde{X}_1$ we express it, using the stabiliser $\tilde{S}_1$ \eqref{eq:stab1} and using the fact that $\tilde{Z}^2_2 = \tilde{Z}^2_3 = \id$, as
\begin{equation}
	\label{eq:X1-S1_identity}
\tilde{X}_1 = \frac{1}{\sin(\theta)} \biggl[\tilde{S}_1 - \cos(\theta)\tilde{Z}_1 \biggl]\tilde{Z}_2\tilde{Z_3},
\end{equation}
and hence subsequently obtain
\begin{equation}
\label{eq:Xsquare}
\tilde{X}^2_1 = \frac{1}{\sin^2(\theta)} \biggl[\tilde{S_1}^2 - \cos(\theta) \big \{\tilde{S_1}, \tilde{Z}_1 \big \} + \cos^2(\theta)\tilde{Z}^2_1 \biggl].
\end{equation}
Using the the expression for $\tilde{S}_1$ \eqref{eq:stab1} and the anticommutation of $\tilde{X}_1$ and $\tilde{Z}_1$ \eqref{eq:anticom} we obtain
\begin{equation}
\begin{aligned}
\big \{\tilde{S_1}, \tilde{Z}_1 \big \} &= \sin(\theta)\{\tilde{X}_1, \tilde{Z}_1 \} \tilde{Z}_2 \tilde{Z}_3 + 2 \cos(\theta) \tilde{Z}^2_1 \\
&= 2 \cos(\theta)\tilde{Z}_1^2,
\end{aligned}
\end{equation}
and so Eq.~\eqref{eq:Xsquare} simplifies to
\begin{equation}
\tilde{X}^2_1 = \frac{1}{\sin^2(\theta)} \biggl[\tilde{S_1}^2 - \cos^2(\theta)\tilde{Z}^2_1 \biggl].
\end{equation}
Finally, because $\tilde{Z}_1^2$ acts as the identity on $\ket*{\tilde\psi}$ \eqref{eq:Z1_id} we have
\begin{equation}
\label{eq:xsquareId}
\tilde{X}^2_1 \ket*{\tilde\psi} = \ket*{\tilde\psi}.
\end{equation}

Now let us demonstrate that for $i=2,3$, the operators $\tilde{Z}_i$ and $\tilde{X}_i$ anticommute when acting on $\ket*{\tilde\psi}$. 
First, we write the products of $\tilde{S}_1$ and $\tilde{S}_i$ (Eqs.~\eqref{eq:stab1}--\eqref{eq:stabi}) as
\begin{equation}
\label{eq:SSi}
\tilde{S}_1 \tilde{S}_i \ket*{\tilde\psi} = \biggl[\sin(\theta) \tilde{X}_1 \tilde{Z}_1 \tilde{Z}_i \tilde{X}_i + \cos(\theta) \tilde{X}_i \tilde{Z}_{5-i} \biggl] \ket*{\tilde\psi},
\end{equation}
\begin{equation}
\label{eq:SiS}
\tilde{S}_i \tilde{S}_1 \ket*{\tilde\psi} = \biggl[\sin(\theta) \tilde{Z}_1 \tilde{X}_1 \tilde{X}_i \tilde{Z}_i + \cos(\theta) \tilde{X}_i \tilde{Z}_ {5-i} \biggl] \ket*{\tilde\psi},
\end{equation}
where we used the facts that $\tilde{Z}_{5-i}^2$ and $\tilde{Z}_1^2$ act as the identity on $\ket*{\tilde\psi}$ (cf.\ Eqs.~\eqref{eq:squarei}--\eqref{eq:Z1_id}). 
Then the two products of $\tilde{X}_i$ and $\tilde{Z}_i$ can be expressed as
\begin{equation}
\tilde{Z}_i \tilde{X}_i \ket*{\tilde\psi} = \frac{1}{\sin(\theta)} \biggl[\tilde{S}_1 \tilde{S}_i - \cos(\theta)\tilde{X}_i \tilde{Z}_{5-i} \biggl] \tilde{Z}_1 \tilde{X}_1 \ket*{\tilde\psi},
\end{equation}
\begin{equation}
\tilde{X}_i \tilde{Z}_i \ket*{\tilde\psi} = \frac{1}{\sin(\theta)} \biggl[\tilde{S}_i \tilde{S}_1 - \cos(\theta)\tilde{X}_i \tilde{Z}_{5-i} \biggl] \tilde{X}_1 \tilde{Z}_1 \ket*{\tilde\psi}.
\end{equation}
From the anticommutation of $\tilde{X}_1$ and $\tilde{Z}_1$ \eqref{eq:anticom} one then obtains
\begin{equation}
\begin{aligned}
\{\tilde{X}_i, \tilde{Z}_i\} \ket*{\tilde\psi} &= (\tilde{X}_i\tilde{Z}_i + \tilde{Z}_i \tilde{X}_i)\ket*{\tilde\psi} \\
& = \frac{1}{\sin(\theta)}\biggl[\tilde{S}_1\tilde{S_i} - \tilde{S}_i \tilde{S}_1 \biggl] \tilde{Z}_1\tilde{X}_1 \ket*{\tilde\psi}.
\end{aligned}
\end{equation}
Using the fact that $\{\tilde{X}_1, \tilde{Z}_1\} = 0$ we can check that $\tilde{Z}_1 \tilde{X}_1$ commutes with both $\tilde{S}_1 \tilde{S}_i$ and $\tilde{S}_i \tilde{S}_1$ and we see that
\begin{equation}
\label{eq:anticomi}
\begin{aligned}
\{\tilde{X}_i, \tilde{Z}_i\} \ket*{\tilde\psi} & = \frac{1}{\sin(\theta)} \tilde{Z}_1\tilde{X}_1  \biggl[\tilde{S}_1\tilde{S_i} - \tilde{S}_i \tilde{S}_1 \biggl]\ket*{\tilde\psi} = 0.
\end{aligned}
\end{equation}

\begin{figure}[t]
     \centering
         \centering
         \includegraphics[width=0.9\textwidth]{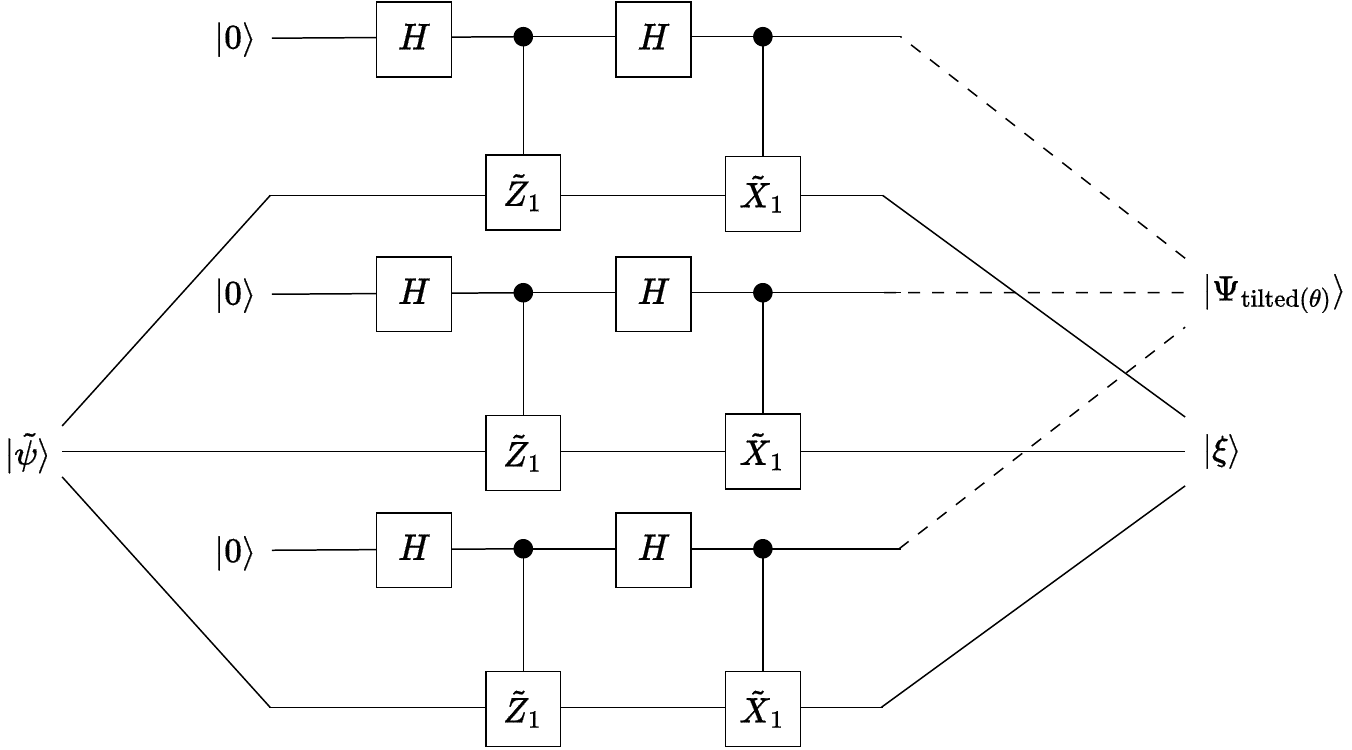}
        \caption{The SWAP isometry used to self test the tilted solution. If the Bell inequality $\mathcal{I}_\theta\le\beta_C$ is maximally violated with the quantum state $\ket{\tilde\psi}$ with quantum value $\beta_Q$, then 
        $\Phi[\ket*{\tilde\psi}]= \ket{\Psi_{\text{tilt}(\theta)}} \otimes \ket{\xi}$. 
        }
        \label{fig:swapIso}
\end{figure}

Let us now conclude the self-testing proof by considering the SWAP isometry depicted in Figure~\ref{fig:swapIso}. 
This isometry makes use of $\tilde{X}_i$ and $\tilde{Z}_i$ as unitary operators, which is not necessarily the case for $\tilde{X}_1$ and $\tilde{Z}_1$ \eqref{eq:1tilde}, as they might have some zero eigenvalues. 
This problem can be easily overcome by regularising these operators, following the method described in Appendix~A.2 of \cite{Supic2020selftestingof}.
Abusing slightly the notation, we redefine in what follows these operators to their regularised versions, $\tilde{X}_1 \leftarrow \tilde{X}_1/|\tilde{X}_1|$ and $\tilde{Z}_1 \leftarrow \tilde{Z}_1/|\tilde{Z}_1|$ for which the zero eigenvalues are changed to $1$, and the operators renormalised. One can verify that those operators act like the non-regularised ones on the state $\ket*{\tilde{\psi}}$.

The output of the SWAP isometry is then 
\begin{equation}
 \Phi[\ket*{\tilde\psi}] = \frac{1}{8} \sum_{i,j,k \in \{0, 1\}} \ket{ijk} \otimes \tilde{X}_1^{i} \tilde{Z}_1^{(i)} \tilde{X}_2^{j} \tilde{Z}_2^{(j)} \tilde{X}_3^{k} \tilde{Z}_3^{(k)} \ket*{\tilde\psi},
\end{equation}
where $\tilde{Z}^{(l)} = \id + (-1)^{l}\tilde{Z}$ and the first three qubits are those of the ancillary systems introduced in the isometry. 
The goal is to recover the tilted state $\ket*{\Psi_{\text{tilt}(\theta)}}$ in the three-qubit system on the left hand-side of the tensor product. 
Let us first consider the case where the first qubit of $\ket{ijk}$ is in the state $\ket{0}$. 
Then for any $j,k \in \{0, 1\}$, 
\begin{align}
\ket{0jk} \otimes \tilde{Z}_1^{(0)} \tilde{X}_2^{j} \tilde{Z}_2^{(j)} \tilde{X}_3^{k} \tilde{Z}_3^{(k)} \ket*{\tilde\psi} =& 
\ket{0jk} \otimes \tilde{Z}_1^{(0)}  \tilde{Z}_2^{(0)} \tilde{X}_2^{j} \tilde{Z}_3^{(0)} \tilde{X}_3^{k} \ket*{\tilde{\psi}}  \notag\\
=& \, (-1)^{j k} \ket{0jk} \otimes
\tilde{Z}_1^{(0)}  \tilde{Z}_2^{(0)} \tilde{Z}_3^{(0)}(\tilde{Z}_1\tilde{Z}_2)^{k} (\tilde{Z}_1\tilde{Z}_3)^{j} \ket*{\tilde\psi} \notag\\
=& \, (-1)^{j k} \ket{0jk} \otimes
\tilde{Z}_1^{(0)}  \tilde{Z}_2^{(0)} \tilde{Z}_3^{(0)} \ket*{\tilde\psi}. 
\end{align}
The first line is obtained by the anticommutation relation \eqref{eq:anticomi} which implies that $\tilde{X}_m^l \tilde{Z}_m^{(l)}=\tilde{Z}_m^{(0)}\tilde{X}_m^l$ for $m=2,3$, the second by the same anticommutation relation and by application of $\tilde{S}_3^k$ and $\tilde{S}_2^j$, and finally the last line uses the fact that $\tilde{Z}^{(0)}_i \tilde{Z}_i \ket*{\tilde\psi} = \tilde{Z}^{(0)}_i \ket*{\tilde\psi}$. 

When the first qubit is in the state $\ket{1}$, we obtain, using an analogous argument, that for any $j,k \in \{0, 1\}$
\begin{align}
\ket{1jk} \otimes \tilde{X_1}\tilde{Z}_1^{(1)} \tilde{X}_2^{j} \tilde{Z}_2^{(j)} \tilde{X}_3^{k} \tilde{Z}_3^{(k)} \ket*{\tilde\psi} =& 
\ket{1jk} \otimes \tilde{Z}_1^{(0)} \tilde{X}_1 \tilde{Z}_2^{(0)} \tilde{X}_2^{j} \tilde{Z}_3^{(0)} \tilde{X}_3^{k} \ket*{\tilde\psi} \notag \\
=& (-1)^{jk}\ket{1jk} \otimes \tilde{Z}_1^{(0)} \tilde{X}_1 \tilde{Z}_2^{(0)}  \tilde{Z}_3^{(0)} (\tilde{Z}_1\tilde{Z}_2)^k (\tilde{Z}_1\tilde{Z}_3)^j \ket*{\tilde\psi} \notag \\
=& (-1)^{j + k + j k}\ket{1jk} \otimes \tilde{Z}_1^{(0)} \tilde{X}_1 \tilde{Z}_2^{(0)}  \tilde{Z}_3^{(0)} \ket*{\tilde\psi}.
\end{align}
Applying the identity~\eqref{eq:X1-S1_identity}, one then obtains
\begin{align}
\ket{1jk} \otimes \tilde{X_1}\tilde{Z}_1^{(1)} \tilde{X}_2^{j} \tilde{Z}_2^{(j)} \tilde{X}_3^{k} \tilde{Z}_3^{(k)} \ket*{\tilde\psi} =& \frac{(-1)^{j + k + j k}}{\sin(\theta)}\ket{1jk} \otimes \tilde{Z}_1^{(0)} \tilde{Z}_2^{(0)}  \tilde{Z}_3^{(0)} \biggl[\tilde{S}_1 - \cos(\theta)\tilde{Z}_1\biggl]\tilde{Z}_2\tilde{Z}_3 \ket*{\tilde\psi} \notag \\
=& \frac{(-1)^{j + k + j k}}{\sin(\theta)}\ket{1jk} \otimes \tilde{Z}_1^{(0)} \tilde{Z}_2^{(0)}  \tilde{Z}_3^{(0)} \biggl[\tilde{S}_1 - \cos(\theta)\tilde{Z}_1\biggl] \ket*{\tilde\psi}  \notag\\
=& \frac{(-1)^{j + k + j k}}{\sin(\theta)} [1 - \cos(\theta)]\ket{1jk} \otimes \tilde{Z}_1^{(0)} \tilde{Z}_2^{(0)}  \tilde{Z}_3^{(0)}\ket*{\tilde\psi}  \notag\\
=& (-1)^{j + k + j k}\frac{\sin(\theta/2)}{\cos(\theta/2)} \ket{1jk} \otimes \tilde{Z}_1^{(0)} \tilde{Z}_2^{(0)}  \tilde{Z}_3^{(0)}\ket*{\tilde\psi},
\end{align}
where the second line is obtained by the commutation of $\tilde{Z}_2$ and $\tilde{Z}_3$ with both $\tilde{S}_1$ and $\tilde{Z}_1$, and the third line using the fact that $\tilde{S}_1$ stabilises $\ket*{\tilde{\psi}}$. 
Finally, the last line is obtained using the trigonometric identities $\sin(\theta) = 2\sin(\frac{\theta}{2})\cos(\frac{\theta}{2})$ and $\cos(\theta) = 1 - 2\sin^2(\frac{\theta}{2})$ . 
With these two expressions the output of the isometry can be written
\begin{align}
\Phi[\ket*{\tilde\psi}] =& \frac{1}{8}\left[ \sum_{j,k \in \{0, 1\}} (-1)^{j k} \ket{0jk} + \frac{\sin(\theta / 2)}{\cos(\theta / 2)} \sum_{j,k \in \{0, 1\}} (-1)^{j+k+j k} \ket{1jk} \right]\otimes \tilde{Z}_1^{(0)} \tilde{Z}_2^{(0)}  \tilde{Z}_3^{(0)}\ket*{\tilde\psi}\notag \\
=& \frac{1}{2}\left[ \cos(\tfrac{\theta}{2}) \sum_{j,k \in \{0, 1\}} (-1)^{j k} \ket{0jk} + \sin(\tfrac{\theta}{2}) \sum_{j,k \in \{0, 1\}} (-1)^{j+k+j k} \ket{1jk} \right] \notag\\
& \quad\otimes \frac{1}{4\cos(\tfrac{\theta}{2})} \tilde{Z}_1^{(0)} \tilde{Z}_2^{(0)}  \tilde{Z}_3^{(0)}\ket*{\tilde\psi} \notag\\
=& \ket*{\Psi_{\text{tilt}(\theta)}} \otimes \ket{\xi},
\end{align}
where $\ket{\xi} = \frac{1}{4\cos(\frac{\theta}{2})} \tilde{Z}_1^{(0)} \tilde{Z}_2^{(0)}  \tilde{Z}_3^{(0)} \ket*{\tilde\psi}$.

We can now prove the self-testing of the observables, i.e., that for all $m=1,2,3$
\begin{equation}
\label{eq:selfTestX}
\Phi[\tilde{X}_m \ket*{\tilde\psi}] = X_m  \ket*{\Psi_{\text{tilt}(\theta)}} \otimes \ket{\xi},
\end{equation}
and
\begin{equation}
\label{eq:selfTestZ}
\Phi[\tilde{Z}_m \ket*{\tilde\psi} ]= Z_m  \ket*{\Psi_{\text{tilt}(\theta)}} \otimes \ket{\xi}.
\end{equation}
Below we show that this is the case for $m=1$ but analogous arguments apply for $m=2,3$.

Reasoning directly, we have
\begin{equation}
\begin{aligned}
\Phi[\tilde{Z}_1 \ket*{\tilde\psi}] =& \sum_{i,j,k \in \{0, 1\}} \ket{ijk} \otimes \tilde{X^i_1}\tilde{Z}_1^{(i)} \tilde{X}_2^{j} \tilde{Z}_2^{(j)} \tilde{X}_3^{k} \tilde{Z}_3^{(k)} \tilde{Z}_1 \ket*{\tilde\psi}\\
=& \sum_{i,j,k \in \{0, 1\}} (-1)^i\ket{ijk} \otimes \tilde{X^i_1}\tilde{Z}_1^{(i)} \tilde{X}_2^{j} \tilde{Z}_2^{(j)} \tilde{X}_3^{k} \tilde{Z}_3^{(k)} \ket*{\tilde\psi} \\
=& \sum_{i,j,k \in \{0, 1\}} Z_1 \ket{ijk} \otimes \tilde{X^i_1}\tilde{Z}_1^{(i)} \tilde{X}_2^{j} \tilde{Z}_2^{(j)} \tilde{X}_3^{k} \tilde{Z}_3^{(k)} \ket*{\tilde\psi} \\
=& Z_1 \ket*{\Psi_{\text{tilt}(\theta)}} \otimes \ket{\xi},
\end{aligned}
\end{equation}
where the second line is obtained using the fact that $\tilde{Z}_1^{(i)} \tilde{Z}_1 \ket*{\tilde\psi} = (-1)^i \tilde{Z}_1^{(i)} \ket*{\tilde\psi}$.
Likewise, we have
\begin{equation}
\begin{aligned}
\Phi[\tilde{X}_1 \ket*{\tilde\psi}] =& \sum_{i,j,k \in \{0, 1\}} \ket{ijk} \otimes \tilde{X^i_1}\tilde{Z}_1^{(i)} \tilde{X}_2^{j} \tilde{Z}_2^{(j)} \tilde{X}_3^{k} \tilde{Z}_3^{(k)} \tilde{X}_1 \ket*{\tilde\psi}\\
=& \sum_{i,j,k \in \{0, 1\}} \ket{ijk} \tilde{X}_2^{j} \tilde{Z}^{(j)}_2  \tilde{X}_3^{k} \tilde{Z}_3^{(k)} \tilde{X}_1^i \tilde{X}_1 (\id + (-1)^{1-i} \tilde{Z}_1)  \ket*{\tilde\psi} \\
=& \sum_{i,j,k \in \{0, 1\}} \ket{ijk} \tilde{X}_2^{j} \tilde{Z}^{(j)}_2  \tilde{X}_3^{k} \tilde{Z}_3^{(k)} \tilde{X}_1^i \tilde{X}_1 (\id + (-1)^{1-i} \tilde{Z}_2 \tilde{X}_3)  \ket*{\tilde\psi} \\
=& \sum_{i,j,k \in \{0, 1\}} \ket{ijk} \tilde{X}_2^{j} \tilde{Z}^{(j)}_2  \tilde{X}_3^{k} \tilde{Z}_3^{(k)} \tilde{X}_1^{1-i} (\id + (-1)^{1-i} \tilde{Z}_2 \tilde{X}_3)  \ket*{\tilde\psi} \\
=& \sum_{i,j,k \in \{0, 1\}} \ket{ijk} \tilde{X}_2^{j} \tilde{Z}^{(j)}_2  \tilde{X}_3^{k} \tilde{Z}_3^{(k)} \tilde{X}_1^{1-i} (\id + (-1)^{1-i} \tilde{Z}_1)  \ket*{\tilde\psi} \\
=& \sum_{i,j,k \in \{0, 1\}} \ket{ijk} \tilde{X}_1^{1-i} \tilde{Z}^{(1-i)}_1 \tilde{X}_2^{j} \tilde{Z}^{(j)}_2  \tilde{X}_3^{k} \tilde{Z}_3^{(k)} \ket*{\tilde\psi} \\
=& \sum_{i,j,k \in \{0, 1\}} \ket{(1-i)jk} \tilde{X}_1^{i} \tilde{Z}^{(i)}_1 \tilde{X}_2^{j} \tilde{Z}^{(j)}_2  \tilde{X}_3^{k} \tilde{Z}_3^{(k)} \ket*{\tilde\psi} \\
=& \sum_{i,j,k \in \{0, 1\}} X_1\ket{ijk} \tilde{X}_1^{i} \tilde{Z}^{(i)}_1 \tilde{X}_2^{j} \tilde{Z}^{(j)}_2  \tilde{X}_3^{k} \tilde{Z}_3^{(k)} \ket*{\tilde\psi} \\
=& X_1 \ket*{\Psi_{\text{tilt}(\theta)}} \otimes \ket{\xi},
\end{aligned}
\end{equation}
where the second line uses the fact that $\tilde{X}_1$ and $\tilde{Z}_1$ anticommute \eqref{eq:anticom}. 
The third line is obtained by application of the stabiliser $\tilde{S}_3$ (\ref{eq:stabi}) (note that for $i=3$, one should use the stabiliser $\tilde{S}_2$) so that we can apply the fact that $\tilde{X}^2_1$ acts as the identity on $\ket*{\tilde\psi}$.

To conclude we show that, for $m=1,2,3$,
\begin{equation}
\Phi[\tilde{X}_m \tilde{Z}_m \ket*{\tilde\psi}] = X_m Z_m \ket*{\Psi_{\text{tilt}(\theta)}} \otimes \ket{\xi}.
\end{equation}
In fact, one can show that a result of this kind holds in general when self-testing observables~\cite{Laura-Private-Comm}, but here we simply show it explicitly for the case at hand. The proof is straightforward from the applications of the two equations Eq.~\eqref{eq:selfTestX} and Eq.~\eqref{eq:selfTestZ}. For the case of $m=1$ (the same arguments applies for $m=2,3$) we have
\begin{equation}
\begin{aligned}
\Phi[\tilde{X}_1 \tilde{Z}_1 \ket*{\tilde\psi}] =& \sum_{i,j,k \in \{0, 1\}} \ket{ijk} \otimes \tilde{X^i_1}\tilde{Z}_1^{(i)} \tilde{X}_2^{j} \tilde{Z}_2^{(j)} \tilde{X}_3^{k} \tilde{Z}_3^{(k)} \tilde{X}_1 \tilde{Z}_1 \ket*{\tilde\psi}\\
=& \sum_{i,j,k \in \{0, 1\}} (-1)^i \ket{ijk} \otimes \tilde{X^i_1}\tilde{Z}_1^{(i)} \tilde{X}_2^{j} \tilde{Z}_2^{(j)} \tilde{X}_3^{k} \tilde{Z}_3^{(k)} \tilde{X}_1  \ket*{\tilde\psi}\\
=& \sum_{i,j,k \in \{0, 1\}} (-1)^i \ket{(1-i)jk} \otimes \tilde{X^i_1}\tilde{Z}_1^{(i)} \tilde{X}_2^{j} \tilde{Z}_2^{(j)} \tilde{X}_3^{k} \tilde{Z}_3^{(k)} \ket*{\tilde\psi}\\
=& \sum_{i,j,k \in \{0, 1\}} X_1 Z_1 \ket{ijk} \otimes \tilde{X^i_1}\tilde{Z}_1^{(i)} \tilde{X}_2^{j} \tilde{Z}_2^{(j)} \tilde{X}_3^{k} \tilde{Z}_3^{(k)} \ket*{\tilde\psi}\\
=& X_1 Z_1 \ket*{\Psi_{\text{tilt}(\theta)}} \otimes \ket{\xi}.
\end{aligned}
\end{equation}
\end{proof}

\end{document}